\documentclass[sn-basic]{sn-jnl}

\jyear{2025}%

\theoremstyle{thmstyleoone}%
\newtheorem{theorem}{Theorem}
\theoremstyle{thmstyletwo}%
\newtheorem{remark}{Remark}%

\theoremstyle{thmstylethree}%
\newtheorem{definition}{Definition}%

\usepackage{graphicx}
\usepackage{enumerate}
\usepackage{natbib}
\usepackage{url}
\usepackage{color}
\usepackage{ulem}
\usepackage{bm}
\usepackage{multicol}
\usepackage{blindtext}

\usepackage[T1]{fontenc} 
\usepackage[utf8]{inputenc} 
\usepackage[english]{babel}
\usepackage{amsmath,amssymb}
\usepackage{mathtools}

\usepackage{microtype} 
\usepackage{rotating} 
\usepackage{booktabs}  
\usepackage{array}     
\usepackage{multirow} 
\usepackage{enumitem}

\newcommand{\1}{\mathbf{1}}

\newcommand{\upp}{{\mathrm{upp}}}
\newcommand{\low}{{\mathrm{low}}}

\newcommand{\lo}{\mathrm{low}}
\newcommand{\up}{\mathrm{upp}}

\newcommand{\RR}{{\rho}}

\begin{document}

\title[Formal Graphical Inference]{
A Formal Graphical Inference Framework for Combining Effect Size with Statistical Significance: Application to Multivariate and Functional Linear Models
}


\author*[1]{\fnm{Tom\' a\v s} \sur{Mrkvi\v cka}}\email{mrkvicka.toma@gmail.com}

\author[2]{\fnm{Mikko} \sur{Kuronen}} \email{mikko.kuronen@luke.fi}

\author[2,3]{\fnm{Mari} \sur{Myllym{\"a}ki}}\email{mari.myllymaki@luke.fi}

\affil*[1]{\orgdiv{Faculty of Agriculture and Technology}, \orgname{University of South Bohemia}, \orgaddress{\street{Studentsk\'a 1668}, \city{\v{C}esk\'e Bud\v{e}jovice}, \postcode{37001}, \country{Czech Republic}}}

\affil[2]{ \orgname{Natural Resources Institute Finland (Luke)}, \orgaddress{\street{Latokartanonkaari 9}, \postcode{FI-00790} \city{Helsinki}, \country{Finland}}}

\affil[3]{ \orgname{Department of Mathematics and Statistics, Faculty of Science, University of Helsinki}, \orgaddress{\street{PL 68, Pietari Kalmin katu 5}, \postcode{FI-00014} \city{HELSINGIN YLIOPISTO}, \country{Finland}}}

\abstract{
To address the critical need for statistical methods that evaluate practical magnitude alongside statistical significance, we introduce a formal graphical inference framework that intrinsically combines effect size with uncertainty. Built upon resampling and global envelopes, this approach provides universality and direct visual interpretability. While existing envelope tests provide only weak family-wise error rate (FWER) or false discovery rate control, we advance the methodology by introducing novel step-down global envelopes. We theoretically prove that several of these envelopes achieve exact strong FWER control under exchangeability, guaranteeing rigorous inference for each individual local hypothesis. Beyond yielding adjusted $p$-values, the framework provides adjusted subset $p$-values for blocks of hypotheses, such as the functional effect of a covariate within a specific category. By explicitly visualizing the size and direction of the effect relative to its variability under the null hypothesis, our approach overcomes the dichotomous nature of traditional testing and provides profound informational value. The proposed framework is applied to multivariate and functional linear models. 
}



\maketitle
\thispagestyle{plain}
\pagestyle{plain}


\section{Introduction}\label{sec:intro}
In response to the American Statistical Association’s warning against relying solely on arbitrary thresholds of statistical significance \citep{WassersteinLazar2016}, researchers have sought more robust metrics to quantify scientific evidence. While effect sizes are increasingly reported to distinguish practical importance from mere statistical detection \citep{SullivanFeinn2012}, other tools address specific inferential gaps; for instance, VanderWeele’s $E$-value assesses the robustness of causal claims against unmeasured confounding \citep{VanderWeeleDing2017}. Concurrently, methodological shifts have introduced Bayes factors to quantify the relative evidence for competing hypotheses \citep{KassRaftery1995} and Second-Generation $p$-values (SGPV) \citep{BlumeEtAl2018}, which incorporate a zone of indifference to avoid trivial significant findings. A distinct paradigm puts forward Shafer’s $E$-value \citep{Shaffer2021} (or betting score), which conceptualizes the test statistic as a bet against the null hypothesis. However, the ''testing by betting'' approach faces criticism regarding its objectivity: the resulting score depends inherently on the specific alternative hypothesis or ''bet'' selected by the researcher \citep{Henning2021}, with different specified alternatives yielding different betting scores for the same data \citep{Dawid2021}. 
The betting score thus lacks the uniqueness of a standard test, complicating its use as a standardized measure of evidence. To resolve the tension between statistical detection and practical magnitude without introducing such ambiguity, a new methodological framework is required—one that intrinsically integrates effect size into the assessment of uncertainty. 

We propose a framework that utilizes the effect size as the primary test statistic, estimating its sampling distribution via Monte Carlo simulation or permutation testing. A key feature of this method is the visualization of the null hypothesis acceptance region alongside the observed effect size. This allows for a comprehensive evaluation of the effect's magnitude relative to the variability under the null hypothesis. 
By doing so, the approach circumvents the limitations of interpreting the $p$-value through a strict dichotomous threshold. 
Instead, it fosters a more nuanced interpretation where the effect size is directly contextualized within the test's acceptance boundaries. 
Consequently, this helps distinguish scenarios where a $p$-value is statistically significant but the underlying effect is negligible — a common pitfall of relying solely on $p$-values — as well as situations where a substantial effect fails to reach significance due to high variability.

Moreover, resampling methods (such as permutation tests) can provide robust alternatives to conventional parametric tests. Unlike parametric methods that rely on asymptotic approximations via the central limit theorem 
or specific distributional assumptions, resampling methods 
derive the null distribution directly through the resampling scheme and 
provide valid inference even with smaller sample sizes or non-standard distributions under an appropriate resampling scheme. The proposed approach is founded on two core principles: Universality, regarding the applicability of the test statistic, and Visual Interpretability, ensuring that the results are intuitive and easily understood. 

The advantages of this formal graphical inference become more pronounced as the test statistic becomes multivariate, functional, 2- or 3-dimensional functional, or even multi-functional. This is demonstrated below with a motivational example. Such settings require rigorous and powerful multiple testing adjustments.

Since our scope is universality, i.e.,  resampling methods are used to address the variability of the test statistic, the single-step or step-down resampling algorithms proposed in \citet{WestfallYoung1993} and \citet{GeEtAl2003} provide a good starting point. Nevertheless, these methods rely on the minimum pointwise resampled $p$-value and thus they can be extremely conservative, especially when computational constraints limit the number of resamples. This can be the case, e.g., when the test statistic is high-dimensional or resamples are computationally expensive. 
As a remedy, in our earlier work, we proposed the single-step global envelope test \citep{MyllymakiEtal2017, MyllymakiMrkvicka2024, MrkvickaEtal2022}, which uses special extremity measures to resolve the problem of ties between the pointwise $p$-values while retaining the intrinsic graphical interpretation (IGI). The IGI property guarantees a one-to-one correspondence between the rigorous test result and the graphical interpretation given by the global envelope (adjusted acceptance region for multiple testing) and the test statistic.

The single-step global envelope test is a global test, i.e., it tests the complete null hypothesis against the alternative that at least one local null hypothesis is false. Such a global test is useful, e.g., in goodness-of-fit testing, when the global hypothesis is of interest. 

When adjusted inference is required for all hypotheses or for a global test of a subset of hypotheses, as in the motivational example below, 
the step-down principle recommended in the literature for strong family-wise error rate (FWER) control is more appropriate.
Therefore, in this paper, we propose the step-down global envelope test as a tool for formal graphical inference. We prove that the proposed procedure controls the FWER strongly, similarly to the step-down resampling minimum $p$-value algorithm. It provides a one-to-one correspondence with the proposed graphical interpretation and satisfies strong FWER control not only for individual hypotheses but also for global tests of subsets of hypotheses. Thus, the proposed algorithm satisfies the Universality and Interpretability properties and is formally correct. 
We propose to use it 
together with the effect size test statistic as a tool to overcome the problems of the dichotomous interpretation of $p$-values in a wide range of applications with univariate, multivariate, functional, or multi-functional test statistics. 
Its applicability is limited only by the existence of an exchangeable or asymptotically exchangeable resampling strategy under the null hypothesis.

The universality and visual interpretability principles can be achieved under the false discovery rate (FDR) control of multiple testing by the so-called false discovery rate envelope \citep{MrkvickaMyllymaki2023}. Nevertheless, we focus here on strong FWER control, as it is one of the most stringent forms of error control in multiple testing. When functional data are under study it is also possible to consider other criteria for controlling false rejections, e.g., interval-wise error rate \citep{PiniVantini2017} and threshold-wise error rate \citep{AbramowiczEtal2023}. 

The proposed step-down global envelope tests are based on the step-down $p$-value algorithm studied in detail in \citet{RomanoWolf2005}. However, the universality principle can also be achieved by other approaches, for example, bootstrap or conformal approaches can be used instead of Monte Carlo methods. The bootstrap allows for approximation of the distribution under the null hypothesis \citep{DavisonHinkley1997}, whereas Monte Carlo tests can provide exact inference \citep{Barnard1963}. The conformal approach \citep{ShaferVovk2008} can be used to test the non-conformity of the data with resampled data. When used in this way, the conformal approach is equivalent to the Monte Carlo test considered in our paper.  

The methods proposed in this article are implemented in the R package GET \citep{MyllymakiMrkvicka2024}. The official version is available from the Comprehensive R Archive Network (CRAN) at \url{https://cran.r-project.org/package=GET}. The methods proposed here are currently included in the development version of
the package, available from \url{https://github.com/myllym/GET}. 

\subsection{Statistics for the effect size determination}

In this article, we apply the proposed framework to multivariate and functional general linear models (GLMs).
To establish a universally interpretable test statistic within this framework, we recommend standardizing both the continuous response and any continuous explanatory variables prior to model fitting. This straightforward transformation is widely advocated in the literature, particularly when formulating models that include interaction terms \citep[e.g.,][]{AikenWest1991, Schielzeth2010}. Consequently, the resulting regression coefficient serves directly as a standardized effect size: a value of $1$ explicitly indicates that a one standard deviation shift in the predictor corresponds to a one standard deviation change in the response. Standardized coefficients are also comparable in magnitude within models as well as between studies. For categorical explanatory variables, although standardized effect size metrics exist (such as Cohen's $d$, Hedges' $g$, or generalized $\eta^2$), we recommend retaining standard contrast coding without further standardization. By employing the model coefficients (or their linear combinations) as the test statistic, the proposed methodology allows for the simultaneous evaluation of all group levels or all post-hoc pairwise differences. The resulting graphical output then intrinsically visualizes the estimated deviations of each group from the overall mean—or the exact differences between specific groups—alongside their rigorous statistical significance. Notably, for a balanced binary predictor (e.g., gender with equal sample sizes), the deviation of a specific category from the overall mean corresponds exactly to a one standard deviation shift in this categorical predictor. Inspired by the input scaling principles introduced by \citet{Gelman2008}, this property ensures that these categorical deviations intuitively maintain a comparable effect size interpretation to the standardized continuous variables. This approach directly preserves the practical interpretability of categorical effects while harnessing the full power of the formal graphical inference.

\subsection{Motivational example}\label{sec:motivational_example}

The advantages of the proposed methodology are illustrated with the following data example. The question of interest is how GDP (gross domestic product) within four regions of the world (vaguely called Continents) has affected population growth over the period 1960-2014. The data come from 112 countries with more than a million inhabitants in 1950 \citep[Figure \ref{fig:populationgrowthdata}, top; source: the World Bank;][]{NagyEtal2017}. That is, we wish to test the interaction effect of the covariates (GDP and Continent) on the response (Growth). 

There are several approaches for statistical inference in the considered linear model (LM) for the functional data. 
A key advantage of the graphical method, the global envelope test, is that it provides rich information, due to the universal choice of the test statistic. Here, the test statistic is the standardized regression coefficient related to the effect of GDP within each continent (Figure \ref{fig:populationgrowthdata}, bottom, black line with red dots). 
The test output shown in Figure \ref{fig:populationgrowthdata} (bottom) gives direct explanations for the rejection of the null model (with $p=0.001$). 
Firstly, it shows significant continents (for all continents $p\leq 0.05$) 
and years (red dots) where the observed regression coefficient leaves the 95\% global envelope (acceptance region; grey band). 
Secondly, the test output provides the direction and size of deviation from the null model, which expresses the observed effect size. This is not provided by $p$-values. 
For example, the observed coefficient exceeds the envelope for the Europe and North America panel (bottom left panel) for years from 1998 onwards, i.e., countries of this continent with higher GDP also had higher growth rates in these years. 
The effect size is about 1, i.e., if GDP is changed by one standard deviation, then the population growth is increased by one standard deviation.
Thus, even though the observed effect is just above the envelope, it is prominent. 

It should be noted that the width of the envelope varies for the different continents and years due to the test being adjusted to the inhomogeneity of the test statistic distribution. No distributional or homogeneity assumptions are required. The necessity of the homogeneity assumption for the usual permutation $F$-max test \citep{WinklerEtal2014} is discussed in the next subsection. 

The weak FWER, under which we developed the global envelope test \citep{MyllymakiEtal2017}, controls the probability of false detection for all hypotheses (here: all years and continents) jointly. Thus, it is suitable for testing whether the multivariate or functional statistic is significant at least for one hypothesis. 
In contrast, here, we work with the strong FWER control, which controls the occurrence of a false detection for any configuration of true and false hypotheses \citep{WestfallYoung1993}, for example for each continent.
Thus, we can trust the significant detections for all four continents (Figure \ref{fig:populationgrowthdata}, bottom).

\begin{figure}[ht!]
    \centering
    \includegraphics[width=0.8\linewidth]{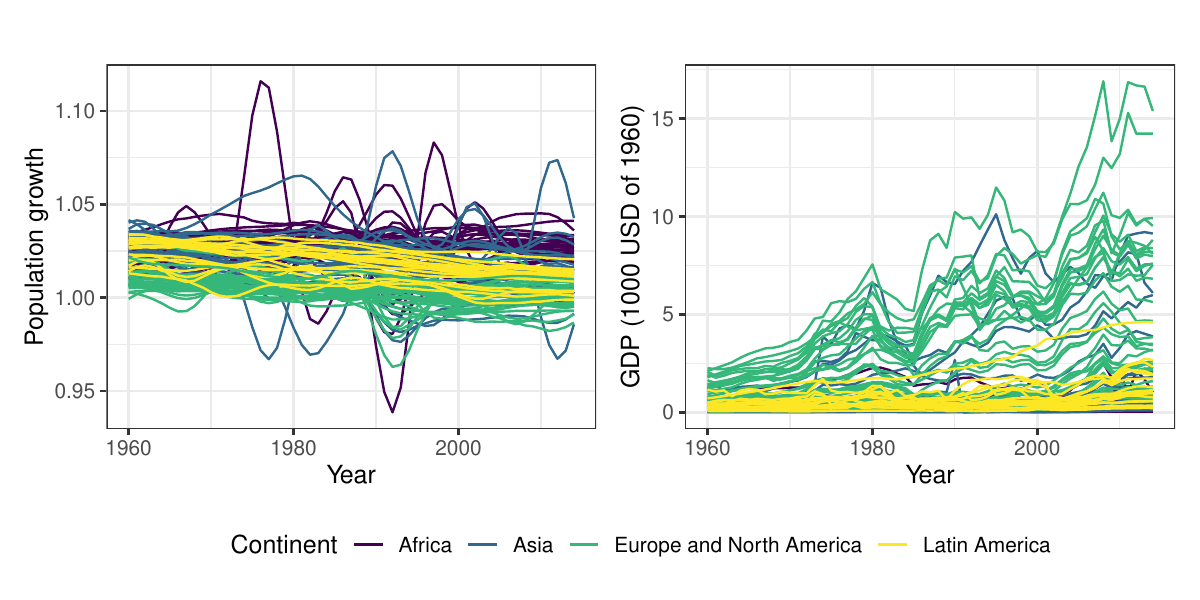}\\
    \includegraphics[width=0.8\linewidth]{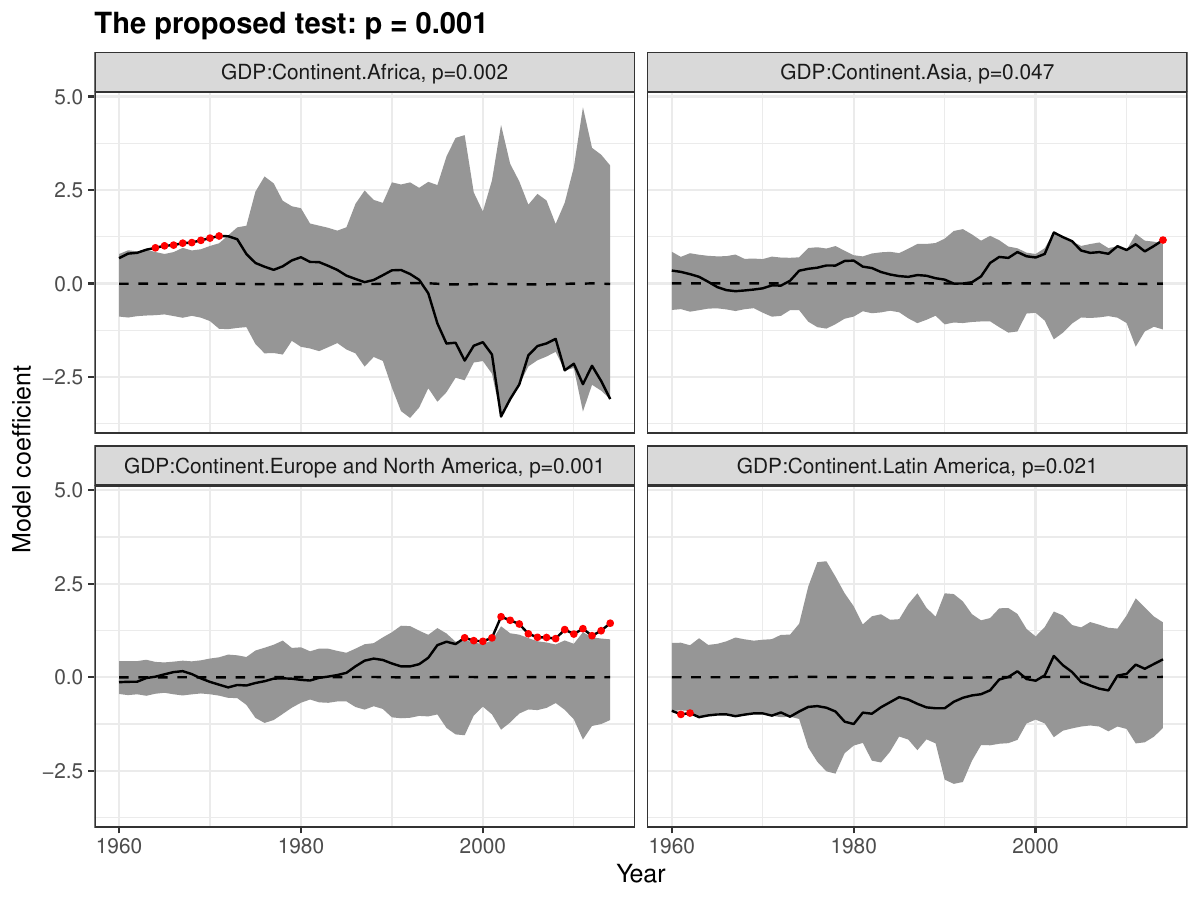}
    \caption{Data (top) and test output (bottom) of the illustrative example. See text for details. 
    GDP is measured in 1960 USD, according to USD inflation, and growth is equal to the population at the end of the year divided by the population at the beginning of the year.  
    Bottom: Black line is the standardized model coefficient estimated from data, grey band is the acceptance region of the test of no joint effect of GDP and Continent on population growth, $p$-values in the panels show the adjusted $p$-values of the respective Continent, which we call adjusted subset $p$-values, and red dots highlight significant years within each Continent.
    }
    \label{fig:populationgrowthdata}
\end{figure}

\subsection{Limitations of traditional maximum-statistic procedures}
\label{sec:fmax_limitations}

\begin{figure}
    \centering
    \includegraphics[width=0.5\linewidth]{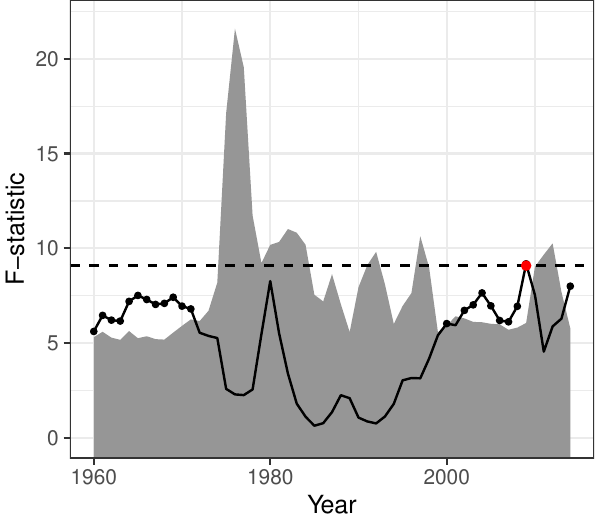}\includegraphics[width=0.5\linewidth]{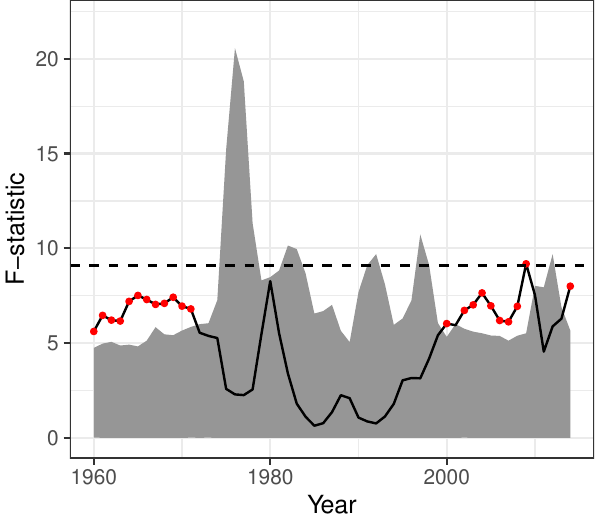}
    \caption{Left: The rank envelope (grey area; corresponding to the pointwise 0.9974 quantile) applied to the $F$-statistics; the dashed line represents the constant $F$-max threshold, and the red point is the only location where $F$-max detects a significant difference.
    Right: The proposed \texttt{sCont} envelope (grey area) applied here to the $F$-statistics, identifying 22 significant points.}
    \label{fig:placeholder}
\end{figure}

In the literature, graphical tests for multivariate or functional data often rely on permutation approaches combined with the pointwise $F$-statistic and the global maximum procedure (e.g., $F$-max) to control the family-wise error rate \citep{WinklerEtal2014}. However, this standard approach has two severe limitations. First, it restricts the user primarily to the $F$-statistic, violating the principle of universality. Consequently, it does not inherently allow for the use of arbitrary test statistics, such as effect sizes, which are crucial for direct, practical interpretations and implicit post-hoc analyses (see the results in Figure \ref{fig:placeholder} and compare them to those in Section \ref{sec:motivational_example}). 

Second, the $F$-max procedure relies heavily on the assumption that the null distribution of the test statistic is relatively homogeneous across the entire functional or multivariate domain. When this assumption fails, the $F$-max test can suffer a substantial loss of power. The motivational example from the previous section which is analysed using the $F$-statistic (Figure \ref{fig:placeholder}) perfectly illustrates this vulnerability. In the left panel, the rank envelope applied to the $F$-statistic (which corresponds to the 99.74\% pointwise quantiles of the $F$-statistic) exhibits extreme heterogeneity, with approximately a 4-fold difference in the envelope size between its narrowest and widest points. 
The $F$-max procedure, illustrated by the dashed line in Figure \ref{fig:placeholder}, imposes a single, constant critical threshold across this entire heterogeneous domain. Although the $F$-max test is globally exact, it achieves this exactness essentially by averaging across the domain. Consequently, in regions where the true pointwise quantiles are substantially lower than the global $F$-max threshold, the test becomes overly conservative. Conversely, in highly variable regions where the true quantiles grow beyond the constant threshold, the $F$-max procedure behaves disproportionately liberally relative to the rest of the domain. 

The practical impact of this limitation is severe: in the presented example, the $F$-max procedure detects only a single significant value (Figure \ref{fig:placeholder}, left panel, red dot). In stark contrast, when applying the proposed step-down continuous rank envelope (sCont) using the $F$-statistic (Figure \ref{fig:placeholder}, right panel), the proposed framework inherently computes the local threshold by constructing the envelope. This robust approach avoids the pitfalls of a constant threshold and successfully identifies 22 significant points while strictly maintaining strong FWER control. 

An example of a mechanism driving another severe inflation of the $F$-statistic in localized regions is detailed in Appendix \ref{app:toy_example} through a controlled simulation scenario.

\section{Background}
\subsection{Multiple hypothesis testing}

Let us assume that we have $m$ tests of null hypotheses $H_{01}, \ldots , H_{0m}$. Table \ref{tab:def} specifies the outcomes arising from $m$ tests. The total number of type I errors among the $m$ tests is denoted by $V$, and $T$ corresponds to the number of type II errors.

\begin{table}
\caption{\label{tab:def}Possible outcomes from $m$ hypothesis tests.}
\centering
\begin{tabular}{cccc} \hline
& Accept & Reject & Total  \\ \hline
True null hypothesis & $U$ & $V$ & $m_0$ \\ 
False null hypothesis & $T$ & $S$ & $m_1$ \\ 
Total & $W$ & $\RR$ & $m$ \\ \hline
\end{tabular}
\end{table}

Let $C$ denote the set of all hypothesis indices, $1, \ldots , m$. For any $K\subset C$, define $H_0^K = \cap_{i \in K} H_{0i}$, i.e., all hypotheses in the subset $K$ are true simultaneously.
The weak FWER  
is defined to be $\mathbb P(V\geq 1 \mid H_0^C)$. The strong FWER 
is $\max_{K\subset C} \mathbb P(V\geq 1 \mid H_0^K)$.

The procedure is said to control weak/strong FWER if its weak/strong FWER is less than or equal to a predetermined significance level $\alpha$. Strong control procedures are usually constructed on the basis of adjusted $p$-values $p_i, i\in C$, that are compared with $\alpha$ and all hypotheses for which $p_i\leq \alpha$ are rejected. On the other hand, they can also be constructed on the basis of a global acceptance region, which we refer here as an envelope. 

We concentrate on resampling-based procedures, where the multivariate/functional test statistic is resampled under the complete null hypothesis. The number of resamples is denoted by $s$. When $s$ increases, we define asymptotic weak/strong FWER control. That is, the procedure is said to control weak/strong FWER asymptotically if its weak/strong FWER is asymptotically less than or equal to $\alpha$.

We will distinguish a global test, i.e., a test of $H_0^C$ against the alternative that at least one null hypothesis is not true, and a local test, i.e., a multiple testing scenario with $m$ hypotheses.

A univariate/multivariate/functional Monte Carlo test is a procedure in which a test statistic $T_0$ taking values in the measurable space $\mathcal S^m$ is resampled $s$ times under the complete null hypothesis, obtaining $T_1, \ldots , T_s$. The measure of extremeness $M$ is computed for every $T_j$ and the null hypothesis is rejected if $M(T_0)$ is among $100\cdot\alpha$\% of the most extreme resamples. If $T_0$ is included in the resampled test statistics in the inference, i.e., if $p=\frac{1}{s+1} \sum_{j=0}^s\1_{M(T_0) \leq M(T_j)}$ (assuming small values of $M$ being extreme), then this Monte Carlo test is exact \citep{Barnard1963, MyllymakiEtal2017}. 
If $T_0$ is not included in the resampled test statistics, then the Monte Carlo test is approximate.
This Monte Carlo test is global and thus it satisfies the weak FWER control for any measure $M$ under the assumption of exchangeability, which means that the joint probability distribution of the variables $T_0, \ldots , T_s$
is permutation invariant, i.e.,
$P\{ (T_0, \ldots, T_s) \in A\}=P\{ (T_{\sigma (0)}, \ldots,  T_{\sigma (s)}) \in A\}$
for any measurable set $A\subset \mathcal S^{m(s+1)}$ and any permutation $\sigma$.

\subsection{Single-step global envelope test with weak FWER control}

The single-step global envelope test defined in \cite{MyllymakiEtal2017} is a global test designed for testing the complete null hypothesis. It uses a Monte Carlo test based on the extremity measure $M$. Five different measures were proposed in \cite{MyllymakiEtal2017} and \cite{MrkvickaEtal2022}: rank, extreme rank length (ERL), cont, area, and qdir. For the sake of completeness, we repeat the definitions here. 

Let us denote the test statistics $T_j = (T_j^1, T_j^2, \dots, T_j^m)$ for $j=0,1,\dots,s$, where the upper index refers to the hypotheses $i=1,\dots,m$.

\subsubsection{Extremity measures}
{\bf Extreme rank (rank)}.
The extreme rank measure of $T_j$ is defined as
\begin{equation}\label{eq:extremerank}
R_j = \min_{i\in C} R_j^i,    
\end{equation}
where $R_j^i$ is the pointwise rank of the test statistic $T_j^i$ among $\{T_0^i, \ldots , T_s^i\}$. The pointwise ranks can be defined as two-sided, decreasing (largest values get smallest ranks) or increasing (smallest values get smallest ranks), according to whether the alternative for $H_i$ is two-sided, greater or less, respectively. The most extreme pointwise rank always obtains the smallest value 1. 

The Monte Carlo test with extreme rank measure corresponds to the well-known $p$-min approach \citep{WinklerEtal2014}. It also corresponds to the single-step resampling algorithm proposed in \cite{WestfallYoung1993} up to the fact that Westfall and Young proposed not to include the data function in the null set and therefore produce an approximate test.

Since $R_j, j=0, \ldots, s$, can have many ties, relying on extreme rank can be rather conservative, especially for big $m$ or small $s$. Therefore, the following three methods for breaking the ties were proposed:
\begin{itemize}
    \item The extreme rank length approach \citep{MyllymakiEtal2017, NarisettyNair2016} takes into account the number of pointwise ranks equalling the extreme rank, 
    \item the continuous rank approach \citep{MrkvickaEtal2022} measures the  extremeness of $T_j^i$ continuously, and 
    \item the area rank \citep{MrkvickaEtal2022} combines the above two approaches.
\end{itemize}

{\bf Extreme rank length (ERL)}.
The ERL measure of $T_j$ is defined based on the vector of the pointwise ordered ranks
$\mathbf{R}_j=(R_j^{[1]}, \ldots , R_j^{[m]})$, where the ranks $R_j^i$, $i=1,\dots,m$, are arranged from smallest to largest.  
The ERL measure of $T_j$ is
\begin{equation}
\label{eq:lexicalrank}
   E_j = \frac{1}{s+1}\sum_{j^\prime=0}^{s} \1(\mathbf{R}_{j^\prime} \prec \mathbf{R}_j)
\end{equation}
where
\[
\mathbf{R}_{j^\prime} \prec \mathbf{R}_{j} \quad \Longleftrightarrow\quad
  \exists \, n\leq m: R_{j^\prime}^{[k]} = R_j^{[k]} \;\forall\, k < n, R_{j^\prime}^{[n]} < R_j^{[n]}.
\]

{\bf Continuous rank (cont)}.
The continuous rank measure of $T_j$ is
\begin{equation}\label{eq:cont}
C_j = \frac{1}{s+1}\min_{i \in C} C_j^i
\end{equation}
where $C_j^i$ are pointwise continuous ranks, and $s+1$ is merely scaling of the values to the interval from 0 to 1.
The definition of 
$C_j^i$ depends on the alternative hypothesis similarly to $R_j^i$:
\begin{equation*}
  C_j^i=\begin{cases}
    c_j^i, &\text{for the one-sided case 'less'} \\
    s-c_j^i, &\text{for the one-sided case 'greater'} \\
   \min(c_j^i, s-c_j^i), &\text{for the two-sided case}.
  \end{cases}
\end{equation*}
where $c_j^i$ is the raw continuous rank of $T_j^i$ among $T_0^i, \ldots, T_s^i$ according to Definition \ref{contrank}.

\begin{definition}\label{contrank} 
Let $y_{[1]} \le y_{[2]} \le \ldots \le y_{[L]}$ denote the ordered set of values $y_l, l=1,2, \ldots, L$.
Define the raw continuous rank such that the smallest $y_l$ has the smallest rank:
\begin{equation*}
c_{[k]} = k - 1 + \frac{y_{[k]}-y_{[k-1]}}{y_{[k+1]}-y_{[k-1]}}
\end{equation*}
for $k = 2, 3, \ldots, L-1$ and
\begin{align*}
c_{[1]} = \exp\left(-\frac{y_{[2]}-y_{[1]}}{y_{[L]}-y_{[2]}}\right) , \quad
c_{[L]} =  L - \exp\left(-\frac{y_{[L]}-y_{[L-1]}}{y_{[L-1]}-y_{[1]}}\right).
\end{align*}
If there are ties, $y_{[l-1]} < y_{[l]} = \ldots = y_{[l']} < y_{[l'+1]}$, then the raw continuous rank is defined as $c_{[k]} = \frac{l+l'}{2} - \frac{1}{2}$ for $k = l, l+1, \ldots, l'$.
\end{definition}

{\bf Area rank (area)}. 
The area rank measure of $T_j$ is
\begin{equation}\label{eq:area_rank_measure}
A_j = \frac{1}{s+1}\left(R_j - \frac{1}{m}\sum_{i=1}^m (R_j-C_j^i) \1(C_j^i < R_j)\right).
\end{equation}

The global $p$-value for the single-step global envelope test with measure $M$ is 
\begin{equation}\label{eq:global_pM}
  p^M=\frac{1}{s+1} \sum_{j=0}^s\1({M_j \leq M_0}).   
\end{equation}
Here $M_j$ stands for $E_j$ when the ERL measure is considered, $C_j$ when the continuous measure is considered, and $A_j$ when the area measure is considered.

{\bf Directional quantile (qdir)}.
The previous four measures are based on the whole non-parametric distributions of the test statistics $T_j^i$ for each $i$. 
However, it is also possible to approximate the distribution by using only a few sample characteristics. 
The directional quantile measure \citep{MyllymakiEtal2015, MyllymakiEtal2017} approximates the empirical distribution using three parameters: the expectation or median $\mu^i$, the upper $\beta$\% quantile $\overline{T}^i$, and the lower $\beta$\% quantile $\underline{T}^i$. It is defined as
\begin{equation}\label{eq:qdir}
D_j = \max_{i=1,\dots,m} \left( \1(T_j^i \ge \mu^i) \frac{T_j^i - \mu^i}{\vert\overline{T}^i - \mu^i\vert} + \1(T_j^i < \mu^i) \frac{\mu^i - T_j^i}{\vert\underline{T}^i - \mu^i\vert} \right).
\end{equation}

For historical reasons, $D_j$ is defined so that it takes larger values for more extreme vectors. 
The $p$-value for the single-step global test with the directional quantile measure is therefore computed differently than for the previous measures, specifically as $p^\text{qdir} = \frac{1}{s+1} \sum_{j=0}^s \1(D_j \ge D_0)$.

Another possible approximation is the studentized measure \citep{MyllymakiEtal2015, MyllymakiEtal2017}, which instead approximates the distribution using the expectation and standard deviation.
When the same distribution of $T_j^i$ is assumed for all $i$,
the unscaled test can be used as well.
The measure is then simply $U_j = \max_{i=1,\dots,m} {\vert T_j^i - \mu^i \vert}$ or, in the case of one-sided (greater) alternative, $U_j = \max_{i=1,\dots,m} T_j^i$. The unscaled test with the $F$-statistic corresponds to the $F$-max test.

\subsubsection{Envelope construction}

For any Monte Carlo test with a measure $M$ defined above, a global envelope can be constructed. A global envelope is a simultaneous band $(T_\text{low}^{\alpha} (i), T_\text{upp}^{\alpha} (i)), i\in C$,
such that the hypothesis $H_0^C$ is rejected if the observed test statistic $T_0^i < T_\text{low}^{\alpha} (i)$ or $T_0^i > T_\text{upp}^{\alpha} (i)$ for any $i\in C$. 
By inspecting $T_0^i < T_\text{low}^{\alpha} (i)$ or $T_0^i > T_\text{upp}^{\alpha} (i)$ for all $i\in C$, the global Monte Carlo test is transformed into a local test, i.e., a multiple-testing procedure. Then, we wish to draw conclusions about any hypothesis $H_{0i}, i\in C$, not only about the global test.

Two kinds of global envelope construction appear in the literature. Firstly, a global envelope was designed for the rank measure \citep{MyllymakiEtal2017}.
The rank envelope was defined by the two boundaries:
$$
T_\low^{\alpha}(i)=\min_{j = 0,\dots,s}^{k_\alpha} T_j^i, \quad  
T_\upp^{\alpha}(i)=\max_{j = 0,\dots,s}^{k_\alpha} T_j^i,
$$
where $\min_{j = 0,\dots,s}^{k_\alpha} T_j^i$ is the $k_\alpha$-th smallest of $T_j^i$, 
$\max_{j = 0,\dots,s}^{k_\alpha} T_j^i$  is the $k_\alpha$-th largest of $T_j^i$, and $k_\alpha$ was set to satisfy
\begin{equation} \label{rankenvelope}
    k_\alpha = \max \left\{k: \sum_{j=0}^{s} \1(R_j<k)\leq \alpha (s+1)\right\}.
\end{equation}

In the multiple-testing scenario, this rank envelope corresponds to the single-step resampling algorithm 2.5 of \citep{WestfallYoung1993}, apart from the inclusion of the data function in the comparison; i.e., $j=0$ is included in the formula above.  
The adjusted $p$-values for the global rank envelope test are then
\begin{equation}
  p^i=\frac{1}{s+1} \sum_{j=0}^s\1({R_j \leq R^i_0})
\end{equation}
for $i\in C$.

These adjusted $p$-values correspond to the graphical interpretation of the envelope, which, for local testing, is as follows: The hypothesis $H_{0i}$ is rejected if $T_0^i < T_\text{low}^{\alpha} (i)$ or $T_0^i > T_\text{upp}^{\alpha} (i)$. \citet{MyllymakiEtal2017} proved one-to-one correspondence between the global $p$-value \eqref{eq:global_pM} that is equal to $\min_{i\in C}p^i$ and the graphical interpretation. The one-to-one correspondence of $p^i$ and the graphical interpretation for $i\in C$ is apparent from their definitions.

The second type of envelopes appeared in \cite{MrkvickaEtal2020, MrkvickaEtal2022}, inspired by functional bands from functional data analysis. Consider a measure $M$ and let $I^\alpha = \{j\in 0,\dots, s: M_j \geq M^{(\alpha)} \}$ be the index set of the test statistics $T_j$ whose measure $M_j$ is larger than the threshold $M^{(\alpha)}$, which is chosen to be the largest value in $\{M_0,\dots,M_s\}$ for which
\begin{equation}\label{eq:Ralpha}
  \sum_{j=0}^{s} \1 \left(M_j < M^{(\alpha)}\right) \leq \alpha (s+1).
\end{equation}
That is, the set $I^\alpha$ contains the indices of the $100(1-\alpha)$\% of the least extreme test statistics $T_j$. The global envelope is defined to be
\begin{equation*}
  T_{\lo}^{\alpha}(i)= \min_{j\in I^\alpha}T_{j}^i, \quad T_{\up}^{\alpha}(i)= \max_{j\in I^\alpha}T_{j}^i.
\end{equation*}

An important feature of the functional band constructed by a functional depth is  \textit{band convexity:} A function with functional depth lower than all functions used to construct the central band must leave the central band for at least one argument $i$ of the function. \cite{NagyEtal2025} showed that the infimal depths constructed on a univariate depth are the only depths satisfying band convexity. In the setting of this paper, the infimal depth corresponds to the rank, cont and qdir measures. The area and ERL measures are  refinements of an infimal depth with respect to ties in finite samples of functions. In other words, ERL and area measures (as well as cont measure) converge to the rank measure when the number of functions increases to infinity.
\citet{MyllymakiMrkvicka2024} called this property in the finite sample setting {\it intrinsic graphical interpretation} (IGI).
Formally, the $100(1-\alpha)$\% global envelope $T_{\lo}^{\alpha}(i), T_{\up}^{\alpha}(i)$ constructed from $T_0, \ldots , T_s$ has IGI with respect to the measure $M$ if
for every $j=0, \ldots , s$, 
\begin{equation}
    T_j^i < T_{\lo}^{\alpha}(i) \text{ or } T_j^i > T_{\up}^{\alpha}(i) \text{ for some } i \in C
\text{ iff  } M_j<M^{(\alpha)}.
\end{equation}
\cite{MrkvickaEtal2022} showed that IGI is satisfied for the rank measure and all its refinements cont, ERL and area, i.e., the one-to-one correspondence between the global test and the graphical interpretation provided by functional band envelopes.

In order to distinguish between the two kinds of envelopes, we will call the second one a functional band envelope and the first one rank envelope. It is important to note that the rank envelope coincides with the functional band envelope based on the rank measure. But it is possible to define an envelope of the rank kind also for cont measure. We did not find this envelope useful for step-down construction and therefore we do not define it here. The continuous rank envelope does not coincide with functional band envelope based on the cont measure.

{\bf Directional quantile envelope (qdir)}.
Just as the directional quantile measure approximates the non-parametric distribution using a few empirical parameters, the directional quantile envelope serves as an approximation of the rank envelope. 

Because the directional quantile measure $D_j$ obtains larger values for more extreme vectors (contrary to the previously defined measures), the threshold $D^{(\alpha)}$ is defined such that it is the smallest value in $\{D_0, \dots, D_s\}$ for which $ \sum_{j=0}^{s} \1 \left(D_j > D^{(\alpha)}\right) \leq \alpha (s+1).$
The $100(1-\alpha)$\% global directional quantile envelope induced by the measure $D_j$ is then defined by the boundaries:
\begin{equation}\label{eq:qdir_env}
    T_{\lo}^{\alpha}(i) = \mu^i - D^{(\alpha)} \left\vert \underline{T}^i - \mu^i \right\vert \quad \text{and} \quad T_{\up}^{\alpha}(i) = \mu^i + D^{(\alpha)} \left\vert \overline{T}^i - \mu^i \right\vert.
\end{equation}

The one-to-one correspondence between the formal global test and this graphical interpretation was proved by \citet{MyllymakiEtal2017}, ensuring that the directional quantile envelope satisfies the IGI property.

\section{Step-down global envelope test with strong FWER control}
In this section, we introduce a step-down envelope construction algorithm based on the step-down procedure by \citet{RomanoWolf2005}. This algorithm computes adjusted $p$-values from a set of pointwise exact $p$-values obtained via resampling. Because the original rank-based approach frequently generates ties among pointwise $p$-values, it typically requires a substantial number of permutations to mitigate conservativeness. To address this limitation, we propose an extension tailored for permutation $p$-values derived from measures that inherently resolve the problem of ties, thereby significantly reducing the required number of permutations. Furthermore, we introduce a step-down FWER envelope that is plotted alongside the test statistic. This simultaneous visualization provides a rigorous graphical interpretation, allowing users to evaluate both the adjusted $p$-values (significance) and the magnitude of the test statistic at a single glance. By establishing a direct, one-to-one correspondence between the adjusted $p$-values and the envelope boundaries, the proposed framework seamlessly integrates statistical significance with effect size.

\subsection{Algorithm}\label{strongenvalgorithm}
\begin{enumerate}
    \item Compute the global envelope test with weak FWER control for $m$ hypotheses and denote 
    the global envelope test $\mathcal T_1$.
    \item Determine the most significant hypothesis $i_1 \in \{1, \ldots , m\}$.
    \item The global $p$-value of $\mathcal T_1$ determines the adjusted $p$-value $p_{i_1}$. The $100\cdot(1-\alpha)\%$ global envelope of $\mathcal T_1$, denoted by $(\widetilde T^{\alpha,\mathcal T_1, M}_\text{low}, \widetilde T^{\alpha,\mathcal T_1, M}_\text{upp})$, determines the strong FWER global envelope at $i_1$, i.e., $T^{\alpha, M}_\text{low}(i_1)= \widetilde T^{\alpha,\mathcal T_1,M}_\text{low}(i_1)$, $T^{\alpha, M}_\text{upp}(i_1)=\widetilde T^{\alpha,\mathcal T_1,M}_\text{upp}(i_1)$. 
    If $p_{i_1}>\alpha$, define the rest of the global envelope by $\mathcal T_1$, i.e., $T^{\alpha, M}_\text{low}(i_j)= \widetilde T^{\alpha,\mathcal T_1,M}_\text{low}(i_j)$, $T^{\alpha, M}_\text{upp}(i_j)=\widetilde T^{\alpha,\mathcal T_1,M}_\text{upp}(i_j)$ for $j=2, \ldots , m$, and stop the algorithm.
    \item Compute the global envelope test with weak FWER control for $m-1$ hypotheses, where the hypothesis $i_1$ is removed, and denote it 
    $\mathcal T_2$. Determine the most significant hypothesis $i_2$ of $\mathcal T_2$. 
    \item The global $p$-value of $\mathcal T_2$ determines the adjusted $p$-value $p_{i_2}$. The $100\cdot(1-\alpha)\%$ global envelope of $\mathcal T_2$, denoted by $(\widetilde T^{\alpha,\mathcal T_2, M}_\text{low}, \widetilde T^{\alpha,\mathcal T_2, M}_\text{upp})$, determines the strong FWER global envelope at $i_2$, i.e., $T^{\alpha, M}_\text{low}(i_2)= \widetilde T^{\alpha,\mathcal T_2,M}_\text{low}(i_2)$, $T^{\alpha, M}_\text{upp}(i_2)=\widetilde T^{\alpha,\mathcal T_2,M}_\text{upp}(i_2)$. 
    If $p_{i_2}>\alpha$ define the rest of the global envelope by $\mathcal T_2$, i.e., $T^{\alpha, M}_\text{low}(i_j)= \widetilde T^{\alpha,\mathcal T_2,M}_\text{low}(i_j)$, $T^{\alpha, M}_\text{upp}(i_j)=\widetilde T^{\alpha,\mathcal T_2,M}_\text{upp}(i_j)$ for $j=3, \ldots , m$, and stop the algorithm.   
    \item Repeat steps 4 and 5 until all hypotheses are removed or the algorithm is stopped. 
\end{enumerate}

\textit{Determination of the most significant hypothesis:}
The step-down algorithm requires identifying the single most significant hypothesis among the remaining unrejected set at each step. This determination is performed as follows:
\begin{itemize}
    \item For the \textit{cont} measure, the most significant hypothesis is uniquely identified as the one with the smallest pointwise continuous rank. For the \textit{qdir} measure, the most significant hypothesis is similarly uniquely identified.
    \item For the \textit{ERL} and \textit{area} measures, the primary selection is inherently determined by the extreme rank $R$. Because the extreme rank frequently generates ties, the tied hypotheses sharing the same extreme rank are decomposed into connected functional or spatial components $\mathcal{R}_r$. For the \textit{ERL} measure, the most significant hypothesis is selected from the largest connected component $\mathcal{R}_r$, and any remaining ties are resolved by centrality (i.e., the hypothesis located most centrally within $\mathcal{R}_r$ is selected). 
    \item For the \textit{area} measure, the connected components $\mathcal{R}_r$ are instead evaluated by their total contribution, defined as $\sum_{i \in \mathcal{R}_r} (R - C_i)$. The most significant hypothesis within the chosen component is then strictly determined by the maximum individual pointwise value $R - C_i$.
\end{itemize}

It is important to emphasize that for the \textit{cont} measure, determining the most significant hypothesis at each step is logically equivalent to establishing the full ordering of significance for all hypotheses at the very beginning of the procedure. Conversely, for the \textit{ERL} and \textit{area} measures, the most significant hypothesis must be evaluated dynamically at each individual step. This step-by-step determination is strictly required to maintain the validity of the one-to-one correspondence between the adjusted $p$-values and the graphical interpretation of the envelopes, as proven in Theorem \ref{11corespondence}.
For the extreme rank measure, all orderings produce equivalent tests.

Appendix \ref{app:implementation} describes an efficient implementation of Algorithm \ref{strongenvalgorithm}.

\subsection{Adjusted $p$-values and adjusted subset $p$-values}

In order to compute adjusted $p$-values for all $i$, the algorithm can be computed for all hypotheses, without stopping. The envelope must be computed with the stopping rule described in Algorithm \ref{strongenvalgorithm} above, only the adjusted $p$-values can be updated. When all the adjusted $p$-values are calculated, it is desirable to apply a monotonicity adjustment, i.e.: if $p_i<p_{i-1}$ for any $i>1$, the adjusted $p_i$ is set to $p_{i-1}$. This guarantees that the $p$-values do not cross the significance threshold after the envelope construction has stopped.

The set of all hypotheses $C$ is often partitioned into $L$ subsets, $C=\cup_{l=1}^L K_l$ and the global $p$-value corresponding to a subset $K_l$ is required. For this purpose we define the adjusted subset $p$-value for a subset of hypotheses $K$ as $\widetilde p^K=\min_{i\in K}p^i$.

\section{Theoretical results}

\begin{definition}
The test statistics $(T_0, \ldots , T_s) \in S^{s+1}$ are called exchangeable if their joint probability distribution 
is permutation invariant, i.e.,
$$Pr\{(T_0, \ldots , T_s) \in A\}=Pr\{(T_{\sigma (0)}, \ldots , T_{\sigma (s)}) \in A\}$$
for any measurable set $A \subset S^{s+1}$ and any permutation $\sigma$.
\end{definition}

Assuming exchangeability under the complete null hypothesis $H_0^C$, 
the Monte Carlo tests based on any measure $M$ that provides a strict ordering of the test statistics are global tests with exact control \citep{Barnard1963, MyllymakiEtal2017} (with the additional condition that $\alpha(s+1)$ must be an integer). 
Thus, in the multiple testing scenario, these tests satisfy the weak FWER control by definition. 

A critical concept for establishing strong FWER control in resampling-based multiple testing is the subset pivotality condition \citep{WestfallYoung1993}. 

\begin{definition}[Subset Pivotality]
The subset pivotality condition holds if the joint distribution of the test statistics for any subset of true null hypotheses $K \subseteq C$ is identical under the full intersection null hypothesis $H_0^C$ and under the partial null hypothesis $H_0^K$.
\end{definition}

In practice, this condition implies that the truth or falsehood of hypotheses outside $K$ does not affect the joint distribution of the test statistics within $K$. When this condition is satisfied, a single, global resampling procedure performed under $H_0^C$ validly simulates the distribution for any subset $K$.

To analyse the behaviour of our proposed measures under subset evaluations, it is useful to categorize them based on how they aggregate information across a given domain. A critical requirement for establishing exact strong FWER control in step-down procedures (without assuming subset pivotality) \citep{RomanoWolf2005} is the structural monotonicity of the measures when stepping down from a larger set to a smaller subset.

\begin{definition}[Pointwise minimum/maximum statistic]
A global measure $M$ is called a pointwise minimum (or maximum) statistic if, for any subset $K \subseteq C$, the global measure over $K$ is defined purely as $M_j = \min_{i \in K} m_j^i$ (or $\max_{i \in K} m_j^i$), where $m_j^i$ is a pointwise value computed solely from the $i$-th component across the simulations $j = 0, \ldots , s$.
\end{definition}

The pointwise continuous rank ($C_j^i$), extreme rank ($R_j^i$), and directional quantile ($D_j^i$) measures strictly satisfy this property. For a restricted subset $K \subseteq C$, their global measures are defined simply as $C_j = \min_{i \in K} C_j^i$, $R_j = \min_{i \in K} R_j^i$, and $D_j = \max_{i \in K} D_j^i$. Because these global measures are derived exclusively through a minimum (or maximum) operation over the subset $K$, restricting the evaluation to a further subset $I \subset K$ structurally guarantees that the new measure will be greater than or equal to (or less than or equal to, for the maximum) the original measure. This intrinsic mathematical property is what naturally enforces the monotonicity of critical values required for strong FWER control.

Conversely, the ERL and area measures do not satisfy this property since their global measures are not constructed via a simple pointwise minimum or maximum. Instead, they evaluate extremity by ranking the entire sub-vectors or functions over $K$. Restricting the evaluation to a smaller subset $I \subset K$ alters the sub-vectors, which can fundamentally change their structural ordering. Consequently, the required monotonicity when moving from $K$ to $I$ is not strictly guaranteed, preventing exact strong FWER control in finite samples.

\begin{remark}
It is important to emphasize the distinction between the underlying pointwise measure and the construction of the functional envelope. While the pointwise continuous rank $C_j^i$ relies exclusively on marginal information from the $i$-th component, the functional band envelope for the cont measure (as well as ERL and area) determines its central index set $I^\alpha$ by evaluating the global extremity of the functions over the entire domain $C$. Consequently, an extreme deviation of a function at a component outside a proper subset $K \subset C$ can cause that function to be excluded from $I^\alpha$. This exclusion subsequently alters the bounding vectors $(T_{\text{low}}, T_{\text{upp}})$ of the envelope within $K$. Because the envelope boundaries within $K$ become dependent on the behaviour of the functions outside $K$, this construction violates the structural isolation required for exact strong FWER control in a single-step procedure.
\end{remark}

\begin{theorem} \label{strongcontrol_singlestep}
Assume the subset pivotality condition holds for the joint distribution of the test statistics 
$T_j=(T_j^i, i \in C)$, $j=0,\ldots, s$. 
The proposed single-step global envelope test achieves exact strong FWER control for the rank and qdir measures $M$.
\end{theorem}

\begin{proof}
\citet{WestfallYoung1993} proved that the single-step resampling procedure achieves strong FWER control under the condition of subset pivotality. This condition ensures that the joint distribution of the test statistics within any subset of true null hypotheses $K \subseteq C$ remains unaffected by the truth or falsehood of hypotheses outside $K$. 

Because the rank and qdir measures are pointwise minimum (or maximum) statistics, their evaluation on any subset $K$, and the construction of their corresponding rank-style envelopes, depends solely on the marginal behaviour within $K$. They are structurally isolated. Consequently, when the underlying data generation mechanism satisfies subset pivotality, this structural isolation ensures that the test itself strictly maintains the pivotal property, guaranteeing exact strong FWER control.

\end{proof}

As noted in the preceding remark, the functional band envelopes for cont, ERL, and area integrate global domain information during the construction of the central set $I^\alpha$. This mechanism breaks the structural isolation. Therefore, even if the underlying test statistics satisfy subset pivotality, the single-step variants of these functional band tests fail to maintain it procedurally, and thus do not strongly control the FWER.

\begin{theorem} \label{strongcontrol_stepdown}
Assume exchangeability of the joint distribution of the test statistics 
$T_j=(T_j^i, i \in C)$, $j=0,\ldots, s$, under the complete null hypothesis $H_0^C$. 
If a global measure $M$ is a pointwise minimum (or maximum) statistic, the proposed step-down algorithm \ref{strongenvalgorithm} achieves exact strong FWER control without requiring the subset pivotality condition. Consequently, the step-down algorithm achieves exact strong FWER control for the rank, cont, and qdir measures. For the ERL and area measures, the algorithm achieves asymptotic strong FWER control as the number of permutations increases.  
\end{theorem}

\begin{proof}
\citet{RomanoWolf2005} demonstrated that step-down procedures guarantee strong FWER control without assuming subset pivotality if the critical values used at each step satisfy a monotonicity condition. Specifically, the critical threshold for evaluating any set of hypotheses $K$ must be at least as stringent as the critical threshold for any of its subsets $I \subset K$. In the absence of specific distributional assumptions, \citet{RomanoWolf2005} enforce this artificially by applying the closure principle.

However, this computationally intensive closure principle is unnecessary if the test measure itself provides the required monotonicity. A pointwise minimum statistic over a set $K$ is computed simply as $\min_{i \in K} m_j^i$. Naturally, the minimum over a large set $K$ is always less than or equal to the minimum over any of its smaller subsets $I \subset K$. (An analogous rule holds for the maximum).

In our step-down procedure, the envelope is recalculated at each step using only the remaining unrejected hypotheses $K$. Because the rank, cont, and qdir measures are pointwise minimum (or maximum) statistics, their critical thresholds $M^{(\alpha)}$ automatically satisfy this monotonicity. Consequently, the procedure directly fulfills the requirement of \citet{RomanoWolf2005} and achieves exact strong FWER control without needing the subset pivotality assumption.

Finally, because the ERL and area measures asymptotically converge to the rank measure, they inherit this structural monotonicity asymptotically, thereby achieving asymptotic strong control within this same step-down framework.
\end{proof}

Theorem \ref{strongcontrol} specifies the strong FWER control of {\it subset global $p$-values}.

\begin{theorem} \label{strongcontrol}
Assume that the domain $C$ is partitioned (or covered) by $L$ subsets, $C=\cup_{l=1}^L K_l$.
Under the assumptions of Theorem \ref{strongcontrol_stepdown}, the proposed step-down algorithm \ref{strongenvalgorithm} provides adjusted $p$-values $p^i$ for each hypothesis $i=1,\ldots, m$. 
We define the subset global $p$-value for the $l$-th subset of hypotheses as $\widetilde p^l=\min_{i\in K_l}p^i$.
The subset global $p$-values achieve exact strong FWER control for the rank, cont, and qdir measures $M$ for every intersection hypothesis $H_0^{K_l}, l=1, \ldots ,L$. 
For the ERL and area measures $M$, the subset global $p$-values achieve asymptotic strong FWER control as the number of permutations increases.  
\end{theorem}

\begin{proof}
By Theorem \ref{strongcontrol_stepdown}, the proposed step-down procedure achieves exact (or asymptotic) strong FWER control over the entire domain $C$. Therefore, the probability of falsely rejecting at least one true null hypothesis $H_{0i}$ is bounded by $\alpha$. Formally, if $K \subseteq C$ denotes the subset of all true null hypotheses, then $\mathbb{P}(\exists i \in K: p^i \leq \alpha \mid H_0^K) \leq \alpha$. 

By definition, the subset global $p$-value satisfies $\widetilde p^l \leq \alpha$ if and only if there exists at least one $i \in K_l$ such that $p^i \leq \alpha$. Consequently, the event that at least one completely true subset hypothesis $H_0^{K_l}$ is falsely rejected (meaning $\widetilde p^l \leq \alpha$ when all hypotheses in $K_l$ are true) implies that at least one true individual hypothesis within those subsets is falsely rejected. However, the converse is not necessarily true, because a true individual hypothesis might belong to a subset that also contains false hypotheses (making the subset hypothesis false overall). Therefore, the probability of falsely rejecting a true subset is bounded by the probability of falsely rejecting a true individual hypothesis:
$$ \mathbb{P}(\exists l \text{ such that } H_0^{K_l} \text{ is true and } \widetilde p^l \leq \alpha \mid H_0^K) \leq \mathbb{P}(\exists i \in K: p^i \leq \alpha \mid H_0^K) \leq \alpha. $$
Thus, the subset global $p$-values achieve strong FWER control for the family of subset hypotheses $H_0^{K_l}, l=1, \ldots ,L$. 
\end{proof}

The following Theorem specifies the one-to-one correspondence of single-step global envelopes.

\begin{theorem} \label{11corespondenceonestep}
Assume that, with probability 1, there are no pointwise ties among $T_j^i$ for $j= 0,\ldots, s$. Then the single-step global envelope test achieves a one-to-one correspondence between the global $p$-value and its graphical interpretation for all considered measures $M$ (rank, cont, ERL, area, and qdir). That is, $p \leq \alpha$ if and only if the observed test statistic falls outside the global envelope for at least one hypothesis. Furthermore, if $p \leq \alpha$, the test statistic strictly exceeds the envelope boundaries at the most significant hypothesis.
\end{theorem}

\begin{proof}
The one-to-one correspondence between the single-step global envelope test and its graphical interpretation was proven in Theorem 1 of \citet{MrkvickaEtal2022} for the cont, ERL, and area measures, and by \citet{MyllymakiEtal2017} for the rank and  qdir measure.

To demonstrate the behaviour at the most significant hypothesis, let $I^\alpha$ denote the index set of the central test statistics used to construct the envelope. For the rank-based measures, there exists a threshold extreme rank $k$ such that $R_j \geq k$ for all $j \in I^\alpha$, and $R_j \leq k$ for all $j \notin I^\alpha$ (equality can occur in both sets because ties at the raw rank $k$ may be split between $I^\alpha$ and its complement by the refined measures). Let $i$ denote the index of the most significant hypothesis. 
If the global test rejects ($p \leq \alpha$), the observed test statistic is not included in the central set ($0 \notin I^\alpha$), which implies its global extreme rank satisfies $R_0 \leq k$. Because $R_0 = \min_{i \in C} R_0^i$, the pointwise rank at the most significant hypothesis $i$ must satisfy $R_0^i \leq k$. Since $R_0^i \leq k$ and $0 \notin I^\alpha$, the observed value $T_0^i$ must fall outside the constructed envelope boundaries $(T^{\alpha}_\text{low}(i), T^{\alpha}_\text{upp}(i))$. An analogous argument holds for the qdir measure using its corresponding threshold $D^{(\alpha)}$ and reversed inequalities.
\end{proof}

The following theorem establishes the one-to-one correspondence for the step-down global envelopes.

\begin{theorem} \label{11corespondence}
Assume that, with probability 1, there are no pointwise ties among $T_j^i$ for $j= 0,\ldots, s$. Then the proposed step-down algorithm \ref{strongenvalgorithm} achieves a one-to-one correspondence between the adjusted $p$-values and the graphical interpretation provided by the envelope for all considered measures $M$ (rank, cont, ERL, area, and qdir). That is, $p_i \leq \alpha$ if and only if the observed test statistic $T_0^i$ falls outside the strong FWER envelope $(T^{\alpha,M}_\low(i), T^{\alpha,M}_\upp(i))$. 
\end{theorem}

\begin{proof}
Consider an arbitrary hypothesis $H_{0i}$ and its corresponding test statistic $T^i_0$. In the step-down procedure, this hypothesis corresponds to some index $i_k$, meaning it is evaluated during the $k$-th subsequent single-step global test, denoted as $\mathcal{T}_k$. 

According to Theorem \ref{11corespondenceonestep}, for the single-step test $\mathcal{T}_k$, the $p$-value satisfies $p_{i_k} \leq \alpha$ if and only if $T^{i_k}_0 \notin (\widetilde T^{\alpha,\mathcal T_k, M}_\text{low}(i_k), \widetilde T^{\alpha,\mathcal T_k, M}_\text{upp}(i_k))$. Because the step-down algorithm defines the final strong FWER envelope at $i_k$ exactly using these single-step boundaries, i.e., $T^{\alpha, M}_\text{low}(i_k) = \widetilde T^{\alpha,\mathcal T_k,M}_\text{low}(i_k)$ and $T^{\alpha, M}_\text{upp}(i_k) = \widetilde T^{\alpha,\mathcal T_k,M}_\text{upp}(i_k)$, the one-to-one correspondence holds by definition for all considered measures $M$ (rank, cont, ERL, area, and qdir).
\end{proof}

\section{Simulation study}\label{sec:simstudy}

We performed a simulation study using a functional ANOVA model to compare the different methods, 
including the single-step and step-down envelope tests with rank, cont, ERL, area, and qdir measures. 
Our tests use the model coefficients as the test statistic, as proposed in this article. 
We also included in the comparisons the most commonly used test for functional GLMs, i.e., the $F$-max test \citep[see, e.g.,][]{WinklerEtal2014}.
The $F$-max test is based on the measure $\max_{i=1,\dots,m} F_i$ where $F_i$ is the $F$-statistic of the $i$th hypothesis. Thus, it is the unscaled test applied to $F$-statistics. 
We used the step-down version of $F$-max according to Algorithm \ref{strongenvalgorithm}.
We denote the step-down tests with the prefix "s", i.e., \texttt{sRank}, \texttt{sCont}, \texttt{sERL}, \texttt{sArea}, \texttt{sqdir}, and \texttt{sFmax}.

We generated data from the functional model 
$Y_{v}(t) = Z_{v}(t) + e(t)$ for $t \in [0,1]$, 
where $Z_v(t)$ specifies the group effect for individual $v$, $v=1,\dots,n$, 
and $e(t)$ is an error term.
For $Z_v(t)$, we considered 10 models.
We adopted the first 10 combinations of group means considered in \citet{MrkvickaMyllymaki2024} as the effects of two groups 1 and 2. 
Model 1 is the complete null case, i.e., $Z_v(t)=0$ for all individuals $v$ and $t\in[0,1]$.
In Models 2-10, $Z_v(t)=0$ for $t\in[0,1]$ if $v$ belongs to group 1, 
but for individuals $v$ of group 2, $Z_v(t)$ deviates from zero as specified in Table \ref{tab:simstudy_means}.
For Model 3, which is presented in the main text below, $Z_v(t) = 1$ for $t\in(0,0.20)$, and 0 otherwise. 
We considered two error distributions:
\begin{enumerate}
    \item Gaussian: $e(t) = \sigma_Z w(t)$, where $w(t)$ is a Gaussian process with the standard deviation $1$ and 
    the exponential covariance function with the scale parameter $\phi_Z = 0$ (independent case), $0.05$, 
    or $0.2$.  
    The values considered for $\sigma_Z$ are given in Table \ref{tab:simstudy_means} (bottom row).
    \item log Gaussian: 
    $e(t) = \sigma_Z \exp(w(t)) / \sqrt{(\exp(1)-1)\exp(1)}$. The resulting marginal distribution had skewness $(\exp(1)+2)\sqrt{\exp(1)-1} \approx 6.18$ and standard deviation $\sigma_Z$. 
\end{enumerate}

We generated functional samples from Models 1-10 at a resolution of $m=200$, using equidistant values of $t$ on $[0,1]$.
For each simulated sample, we considered the model 
$$
\mathbf{E}[Y_{v}(t)] = \beta_0(t) + \beta_1(t) \cdot \mathbf{1}(v \in \text{group } 2) 
$$
for all $m=200$ values of $t \in [0,1]$.
We generated resamples under the null hypotheses of no group effect by permuting the functional data across groups. For the considered ANOVA model, the resamples are exchangeable under the null hypothesis and the tests should therefore provide either exact or asymptotic strong FWER control. 
We performed the tests 
with estimated $\beta_1$ coefficients (all envelope tests) or $F$-statistics (\texttt{Fmax}, \texttt{sFmax}) as test statistics.

We wished to show: 
i) the methods empirically control FWER (for \texttt{sERL} and \texttt{sArea} asymptotically only), providing support to the theoretical results, 
ii) our proposed tests, which have greater graphical interpretation, have similar, or even higher, power than $F$-max, and
iii) how the proposed methods perform relative to each other.
First, both for the Gaussian and log Gaussian errors, we varied the number of observations $n_1$ of group 1 from 10 to 100, keeping the number of observations in the second group fixed to $n_2=3 n_1$. 
For each test, we used $s=999$ permutations.
Second, to show the behaviour of empirical FWER as $s$ increases, particularly for \texttt{sERL} and \texttt{sArea}, we considered the case $n_1=30$ and $n_2=90$ with $s=99, 199, 399, 599, 799, 999, 1999, 2999, 3999$ or $4999$ resamples.
Third, to show sensitivity of \texttt{sFmax} in power to imbalanced data, we varied $n_2$ from 10 to 3981 keeping $n_1$ fixed to $10$, while using $s=999$ permutations.
In this third case, we used $2/3 \sigma_Z$ instead of $\sigma_Z$ to have a reasonable power range.
For each considered model, we generated 1000 functional samples. 
In all cases, we considered the nominal $\alpha=0.05$. 
We inspected the control of FWER as well as power.
In this work, power is defined as the proportion of the alternative hypotheses rejected.
For concise presentation, we present only key observations in the main text (from Model 3 with $\phi_Z = 0.05$), while the complete set of results can be found in Appendix \ref{suppl:simstudy} (Figures \ref{fig:simstudy_fwer_Gaussian}-\ref{fig:simstudy_power_logGaussian2}).

All the methods controlled FWER for Model 1 (Figures \ref{fig:simstudy_fwer_Gaussian} and \ref{fig:simstudy_fwer_logGaussian}).
In Model 1, we have no $H_1$ hypotheses and are testing the weak control (the global test across all hypotheses).
For Models 2-10, we see that
\texttt{rank}, \texttt{qdir}, \texttt{sRank}, \texttt{sCont}, \texttt{sqdir} provide strong control,
while \texttt{sERL} and \texttt{sArea} do not provide such control for all $s$ (Figure \ref{fig:simstudy_fwer}, and \ref{fig:simstudy_fwer_Gaussian}-\ref{fig:simstudy_fwer_n}). 
However, \texttt{sERL} and \texttt{sArea} reach the desired FWER level with increasing $s$, which supports the asymptotic control results (Figures \ref{fig:simstudy_fwer} right, and \ref{fig:simstudy_fwer_n}).
We note also that the \texttt{sRank} test was very conservative.

Secondly, we inspected the power for measures \texttt{Fmax}, \texttt{sFmax}, \texttt{qdir}, \texttt{sqdir}, and \texttt{sCont}, which
control FWER strongly. All the methods had rather similar power in the Gaussian case (Figure \ref{fig:simstudy_power} left, and \ref{fig:simstudy_power_Gaussian}). 
Here, the $F$-statistic distribution is homogeneous across $t\in[0,1]$, i.e., this case is ideal for \texttt{sFmax}. 
The log Gaussian case represents an example where the $F$-statistic distribution is not homogeneous. Here, \texttt{sFmax} has clearly lower power than \texttt{qdir}, \texttt{sqdir} and \texttt{sCont} for all sample sizes $n_1$ and $n_2=3 n_1$ (Figure \ref{fig:simstudy_power} middle, and \ref{fig:simstudy_power_logGaussian}), and dramatically lower power under increasing imbalance (Figure \ref{fig:simstudy_power} right, and \ref{fig:simstudy_power_logGaussian2}). 
Comparing the proposed methods against each other, we see:
1) \texttt{sFmax} and \texttt{sqdir} have a bit higher power than \texttt{Fmax} and \texttt{qdir}, respectively, and
2) \texttt{sCont} is the most robust alternative across all studied cases.

\begin{figure}
    \centering
    \includegraphics[width=\linewidth]{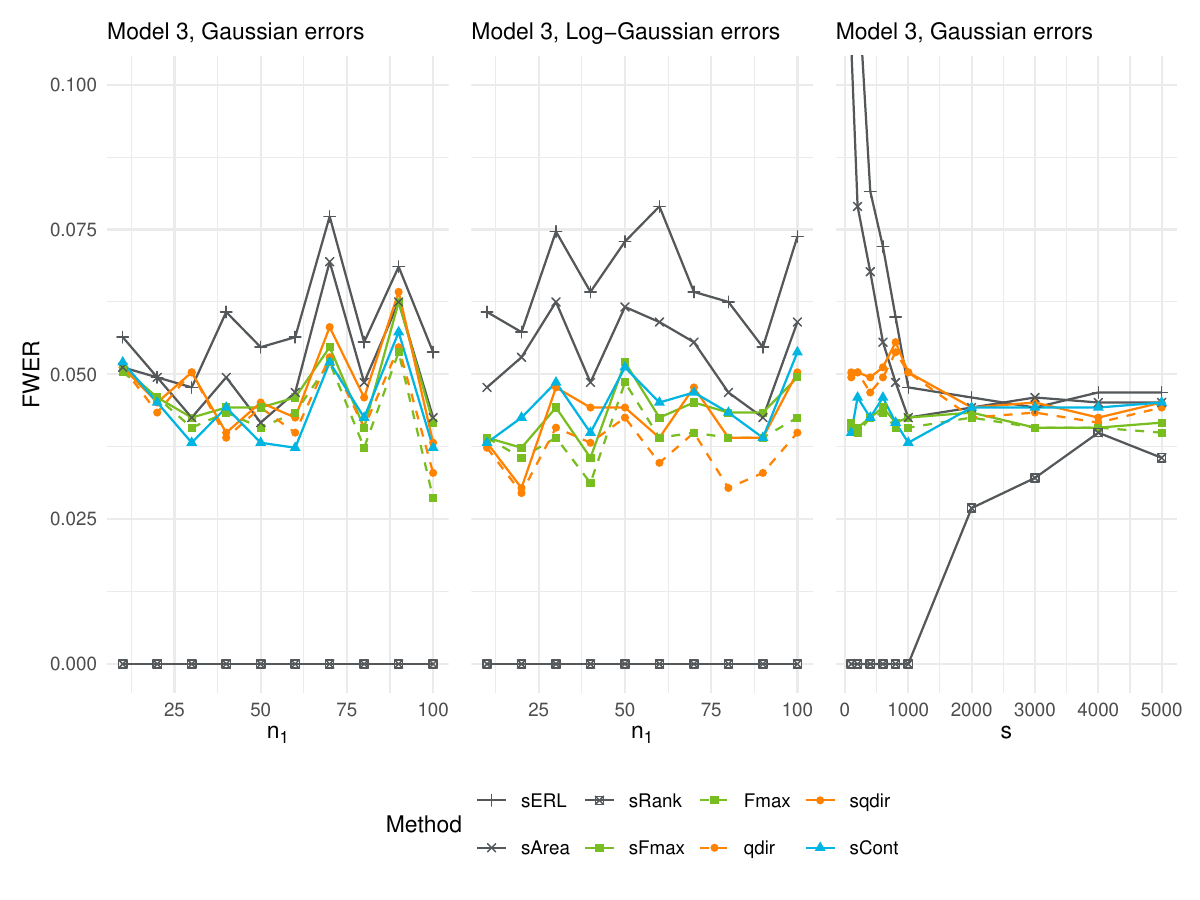}
    \caption{Empirical FWER 
    for the different methods (curves) for Model 3 with Gaussian or log-Gaussian errors as given in the headings, 
    either with respect to increasing 
    numbers of observations in the first group, $n_1$, keeping the number of simulations in the second group fixed to $n_2=3 n_1$ (i.e., also increasing) and the number of simulations $s=999$ (left \& middle),  
    or with respect to varying $s$ with fixed $n_1=30$ and $n_2=3 n_1$ (right). 
    In all cases, $\phi_Z = 0.05$, and $m=200$. 
    }
    \label{fig:simstudy_fwer}
\end{figure}

\begin{figure}
    \centering
    \includegraphics[width=\linewidth]{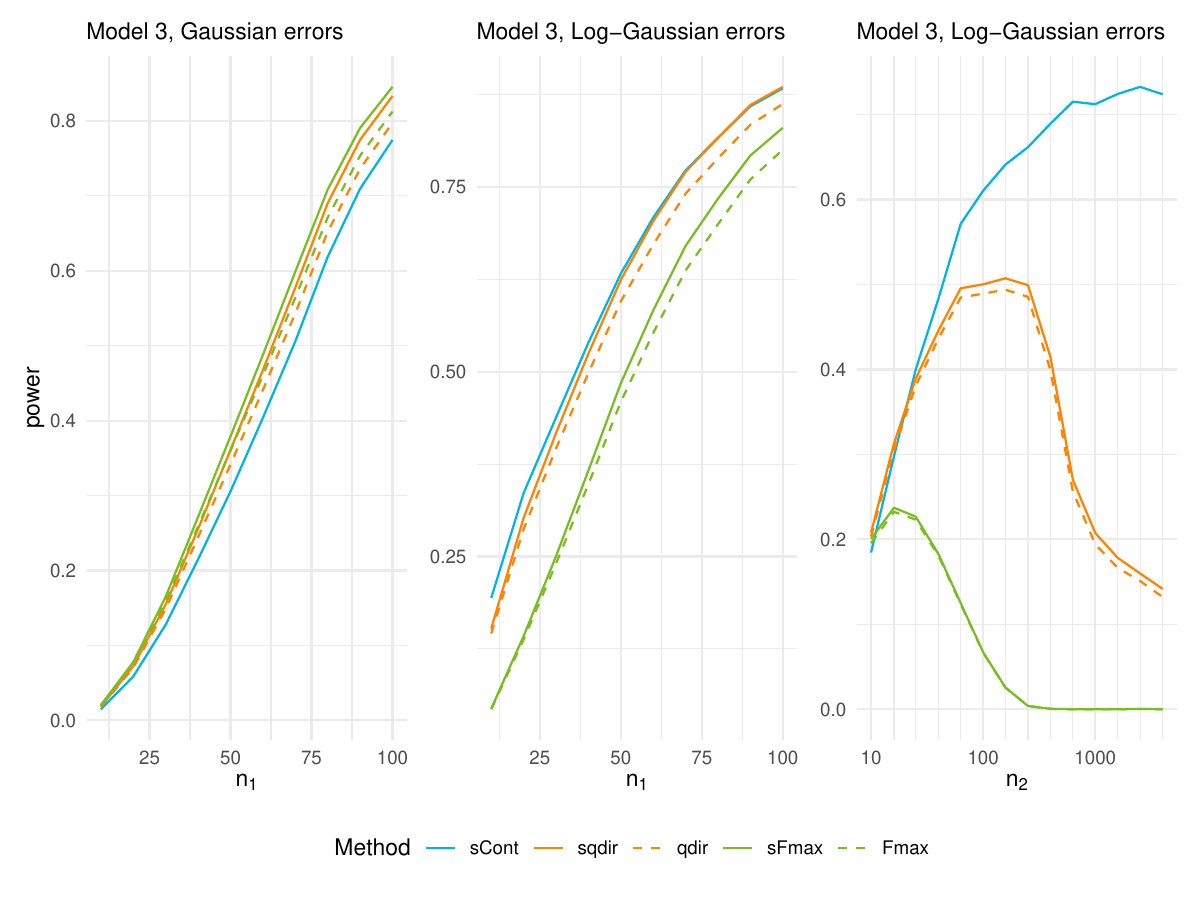}
    \caption{Empirical power (i.e., mean proportion of the alternative hypotheses rejected) for the different methods (curves) for Model 4 with Gaussian or log-Gaussian errors as given in the headings, 
    either with respect to varying numbers of observations in the first group, $n_1$, with $n_2 = 3 n_1$ (left \& middle), 
    or in the second group, $n_2$ with fixed $n_1=10$ (right).
    In all cases, $\phi_Z = 0.05$, $s=999$, and $m=200$.
    }
    \label{fig:simstudy_power}
\end{figure}

\section{Data study using a multivariate linear model}

The proposed formal graphical inference can also be used for individual coordinate testing in a multivariate linear model.
We show here an example study on how mixed forest stands' soil properties are affected by nearby tree species.
We repeat the analysis of \citet{SoronenEtal2025}, who conducted their analyses under FDR control, with our new proposed methods.
The nine soil properties considered are
carbon stock (C), 
soil nitrogen stock (N), and
C:N ratio in the forest floor (FF), 
these same properties (C, N, C:N) in mineral soil (0–10 cm depth; MS), 
as well as the fluxes of ammonium-N (NH4-N), nitrate-N (NO3-N) and 
the total amino acid-N (AA-N). 
The properties were measured near 17 alder, 14 birch and 17 spruce trees in Norway spruce dominated sites in southern Finland.
For details of the study sites and measurements made, see \citet[][incl.\ Table 1 for a study site summary]{SoronenEtal2025}.

We used the multivariate linear model of \citet{SoronenEtal2025}:
soil properties $\sim$ species + fertility + site + age, where species (alder, birch, spruce) is the interesting covariate, and fertility \citep[two classes, see][]{SoronenEtal2025}, 
site, 
and age of the stand are (fixed) nuisance covariates.
For each soil property, we have the same null hypothesis: the tree species has no effect on the soil property, i.e., 
\begin{equation}\label{eq:forest_H01}
H_{0i}: \beta_{\text{Alder},i}=\beta_{\text{Birch},i}=\beta_{\text{Spruce},i}=0 
\end{equation}
for all soil properties $i=1,2,\dots,9$.
Here $\beta_{\text{Alder},i}$, $\beta_{\text{Birch},i}$ and $\beta_{\text{Spruce},i}$ are elements of the $\beta$-coefficient vector related to the species factor in the linear model for the $i$th soil property.

To compare the magnitude of effects across different soil properties, 
we first show the test based on standardized coefficients, allowing comparison of effects across soil properties and tree species. 
Organizing the coefficients by soil property, our test statistic is
\begin{equation}\label{eq:forest_teststat1}
T = (\beta_{\text{Alder},1}, \beta_{\text{Birch},1}, \beta_{\text{Spruce},1}, 
\ldots, 
\beta_{\text{Alder},9}, \beta_{\text{Birch},9},
\beta_{\text{Spruce},9}).    
\end{equation}
The results show the effect of each species in comparison to the mean effect for each soil property (Figure \ref{fig:mult_example_forest_means}). 
The panels of Figure \ref{fig:mult_example_forest_means} show global soil-property-specific $p$-values, with statistically significant effects for N and C:N in FF, and C:N in MS.
More specifically, alder affected N in FF positively, and C:N negatively,
while spruce had a positive effect on C:N.
All these standardized effects are about 0.6 in absolute size. 
This means, e.g., that having a nearby alder tree increased N in FF by 0.6 standard deviation relative to the mean N.
There was also a positive spruce effect of 0.4 on the C:N ratio.

The hypothesis \eqref{eq:forest_H01} can be equivalently expressed as $H'_{0i}$: $\beta_{\text{Alder},i}-\beta_{\text{Birch},i}=0$, $\beta_{\text{Alder},i}-\beta_{\text{Spruce},i}=0$ and $\beta_{\text{Birch},i}-\beta_{\text{Spruce},i}=0$.
To test this hypothesis specifically, the test statistic can be taken to consist of the differences between $\beta$-coefficients related to species, as in \citet{MrkvickaEtal2021a} for global testing. Organizing the coefficients by soil property,
$T = (\beta_{\text{Alder},1}-\beta_{\text{Birch},1}, \beta_{\text{Alder},1}-\beta_{\text{Spruce},1}, \beta_{\text{Birch},1}-\beta_{\text{Spruce},1}, \ldots, 
\beta_{\text{Alder},9}-\beta_{\text{Birch},9}, \beta_{\text{Alder},9}-\beta_{\text{Spruce},9}, \beta_{\text{Birch},9}-\beta_{\text{Spruce},9})$.
This test statistic, which was also used by \citet{SoronenEtal2025}, allows direct detection of differences between species.
In this case, it is reasonable to inspect differences in the original units, while differences in standardized units are less interpretable.
The test output (Figure \ref{fig:mult_example_forest_contrasts}) shows that tree species had a significant effect on N and C:N ratio in FF, similarly as did the test based on the test statistic \eqref{eq:forest_teststat1}.
Here, the results indicate that N was on average about 0.045 kg/m$^2$ higher near alders than near birches and 
0.03 kg/m$^2$ higher near alders than near spruces, and 
that the C:N ratio was about five units lower near alders than near spruces, 
after adjusting for the effects of the other factors. 
However, in contrast to the results with the test statistic \eqref{eq:forest_teststat1}, using the statistic of differences C:N in MS was not statistically significant ($p=0.0858$). The test output, however, shows that the differences between alder and spruce, as well as between birch and spruce are borderline significant (close to the envelope boundary), and the C:N ratio was 5 units smaller near alder and birch than near spruce. 
Compared with the results of \citet{SoronenEtal2025}, our analysis uses stronger control, namely strong FWER instead of FDR control. The effects on C:N in MS and a difference between alder and birch in C were not statistically significant, although they were borderline cases with the same sizes of effects as reported by \citet{SoronenEtal2025}.

\begin{figure}
    \centering
    \includegraphics[width=1\linewidth]{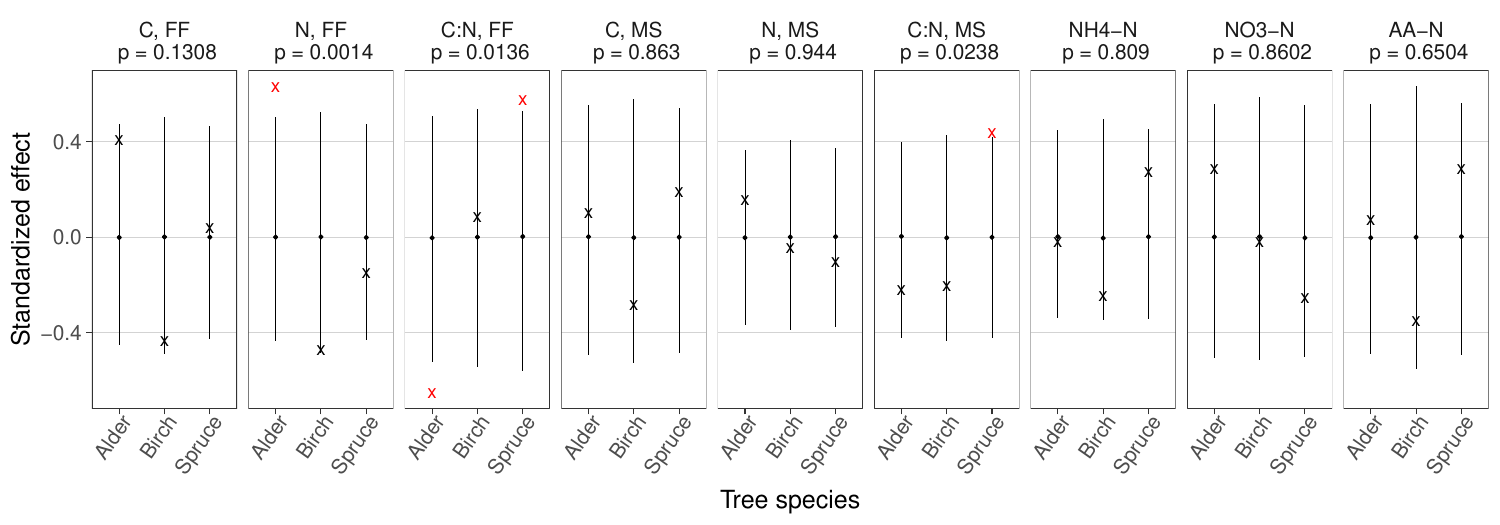}
    \caption{The test output for the hypotheses \eqref{eq:forest_H01} of no effect of tree species on the soil properties using the test statistic \eqref{eq:forest_teststat1}. 
    Panels show the mean standardized effects for each tree species and soil property (FF = forest floor, MS = mineral soil 0-10 cm depth), after accounting for fertility, site and age.
    The crosses are the observed standardized effects. 
    The whiskers represent the acceptance regions of the tests. 
    The $p$-values in the panels are the global soil-property-specific $p$-values, and red crosses indicate the significant tree species effects within each soil property.}
    \label{fig:mult_example_forest_means}
\end{figure}

\begin{figure}
    \centering
    \includegraphics[width=1\linewidth]{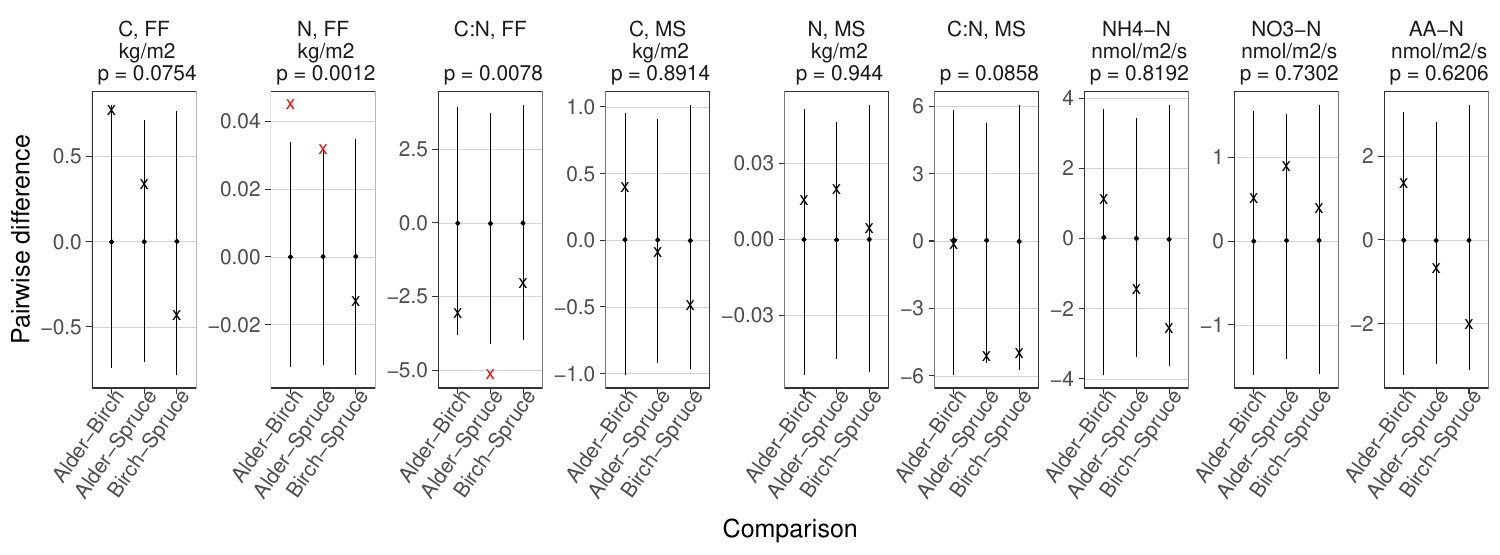}
    \caption{The test output of testing hypotheses of no differences between the tree species effects on the soil properties using the difference-based test statistic. 
    Panels show the mean difference between two species' effects for each pairwise species comparison and each soil property (FF = forest floor, MS = mineral soil 0-10 cm depth), after accounting for fertility, site and age.
    The crosses show the observed pairwise differences between the species.
    The whiskers represent the acceptance region of the test.
    The $p$-values in the panels are the global soil-property-specific $p$-values, and red crosses indicate the significant pairwise differences between tree species effects within each soil property.}
    \label{fig:mult_example_forest_contrasts}
\end{figure}

\section{Discussion and Conclusions}\label{sec:discussion}

In this article, we introduced the step-down global envelope test, a novel framework for formal graphical inference that achieves strong family-wise error rate (FWER) control. By combining the universality of resampling methods with the rich interpretability of global envelopes, 
the framework enables statistical significance and effect size to be assessed simultaneously.

Among the evaluated measures, the continuous step-down measure (\texttt{sCont}) provides the most robust and reliable choice. Under exchangeability, it guarantees strong FWER control for any finite number of permutations. The \texttt{sqdir} method also provides exact strong FWER control but relies on a three-parameter approximation of the test statistic's distribution. It was less robust than \texttt{sCont} in terms of power across different simulation scenarios. 
The \texttt{sERL} and \texttt{sArea} methods achieve strong FWER control only asymptotically. Furthermore, our toy example and simulation study demonstrate an important limitation of the traditional $F$-max procedure: because the $F$-max test relies on the assumption of a homogeneous null distribution across the hypotheses, it can suffer from a substantial loss of power when the distribution is heterogeneous, as demonstrated here for non-Gaussian data combined with imbalanced designs—a setting that can commonly arise in applications of permutation tests. 

The primary methodological contribution of this paper is the extension of global envelopes to step-down procedures. While we demonstrated that step-down methods can improve power over their single-step counterparts, their main advantage lies in local inference. 
Applied researchers rely heavily on local interpretations (e.g., identifying specific time points or spatial regions). 
Previous single-step envelopes, which primarily provided weak FWER control, could exhibit liberal behaviour for local hypotheses, particularly in borderline cases. 
The proposed step-down \texttt{sCont} method resolves this fundamental issue by providing exact strong FWER control: researchers can interpret each local exceedance of the envelope, knowing that the strong FWER is strictly controlled.

However, we do not dismiss single-step tests or the \texttt{sERL} and \texttt{sArea} methods. If the primary research question concerns the global hypothesis rather than the localization of individual deviations, the \texttt{sERL} and \texttt{sArea} tests can be particularly useful. By accumulating information from neighbouring hypotheses, the ERL and area measures can provide high power for the global test. Their step-down versions, \texttt{sERL} and \texttt{sArea}, additionally reduce the local liberality of their single-step counterparts, while achieving strong FWER control asymptotically.

A distinct advantage of our universal framework is the ability to directly use effect sizes—such as standardized regression coefficients—as the primary test statistic. This shifts the paradigm from merely detecting significance to evaluating statistical significance together with the practical magnitude of the effect. 
For categorical variables, the methodology accommodates different parametrizations: researchers can visualize contrasts to the overall mean, evaluate all post-hoc pairwise differences simultaneously, or compute contrasts against a specific reference category. The intuitive visual nature of these effect-size envelopes makes them particularly attractive for applied research, where accessible interpretation is important alongside rigorous statistical control.

While the present study focused on one-dimensional functional data, the methodology naturally extends to higher dimensions. For 2D functional data (e.g., spatial domains), envelope exceedances can be visualized by plotting circles whose radii correspond to the magnitude of the deviation \citep{MyllymakiMrkvicka2024}. For 3D data, the results can be visualized through a series of 2D slices. Moreover, the flexibility of the global envelope allows for the concatenation of multiple functional test statistics, enabling 
simultaneous testing of different aspects of functional deviations or evaluation of multiple distinct test statistics within a single, unified graphical output.

The methods proposed herein will be fully implemented in the R package \texttt{GET} \citep{MyllymakiMrkvicka2024}. In addition to the strong and weak FWER controls discussed in this paper, the package also offers envelopes designed to control the false discovery rate (FDR). Our ongoing research aims to extend the linear modelling framework to accommodate mixed-effects models and heteroscedastic errors. While these extensions require different and more complex permutation strategies, the same general construction and graphical interpretation of the step-down global envelopes proposed in this paper remain applicable.

\section*{Acknowledgements}

The authors thank Päivi Soronen (Luke) for providing the soil properties data.
MK and MM conducted the study within the Research Council of Finland's flagship ecosystem for Forest-Human-Machine Interplay---Building Resilience, Redefining Value Networks and Enabling Meaningful Experiences (UNITE) (Grant numbers 359174 and 374286). 

\vspace{1em}
\noindent \textbf{Disclosure regarding AI use:} The authors utilized a generative AI assistant solely for language refinement and copy-editing to ensure standard academic English. All methodological innovations, statistical implementations, and data interpretations remain the independent work of the authors. The authors have reviewed and assume full responsibility for the content of the final manuscript.

\bibliography{references_cited}

\appendix

\setcounter{figure}{0}
\setcounter{table}{0}
\setcounter{equation}{0}
\renewcommand{\thefigure}{A\arabic{figure}}
\renewcommand{\thetable}{A\arabic{table}}
\renewcommand{\theequation}{A\arabic{equation}}
\renewcommand{\theHfigure}{A\arabic{figure}} 
\renewcommand{\theHtable}{A\arabic{table}} 
\renewcommand{\theHequation}{A\arabic{equation}}

\section{A toy example that produces highly heterogeneous $F$-statistics }\label{app:toy_example}

The severe heterogeneity of the $F$-statistic under permutation, which critically undermines the standard $F$-max procedure as discussed in Section \ref{sec:fmax_limitations}, can appear for various reasons. Here we present a simple toy example that demonstrates it for log-normal data containing regions with balanced and unbalanced designs. Such a design naturally arises in many scientific applications where the covariate of interest is location-dependent across the investigated domain (e.g., localized presence of a brain lesion in neuroimaging, spatially varying genetic mutations in genomics, or time-varying treatment of plants).

To vividly illustrate this mechanism, consider a simplified toy scenario with a skewed (log-normal) response and $N=16$ observations. Suppose a functional domain contains regions with identical theoretical degrees of freedom for an $F$-test ($F_{1, 14}$) but varying levels of design imbalance: a perfectly balanced Region A (8 vs. 8), a moderately imbalanced Region B (14 vs. 2), and an extremely imbalanced Region C (15 vs. 1). 

Under permutation, extreme values in the data interact very differently with the design matrix depending on the group sizes. In imbalanced designs, observations in the smaller group carry disproportionately high leverage. When large values from the skewed response are permuted into these very small groups, the model absorbs much of the variance, artificially deflating the residual sum of squares and causing the pointwise $F$-statistic to inflate massively. 

Consequently, even under the same null hypothesis, the empirical permutation distributions become heavily heterogeneous across the domain. According to our simulations for this scenario, the $99\%$ quantile of the permuted $F$-statistic in the moderately imbalanced Region B is 3.4 times larger than in the perfectly balanced Region A. In the extremely imbalanced Region C, this discrepancy increases to an approximately 23-fold difference.

If a global maximum statistic ($F$-max) is used to control the strong FWER across this domain, the constant critical threshold will be entirely dictated by the extreme values produced in the highly imbalanced Region C. This artificially high, constant threshold renders the test virtually blind to any genuine effects in the stable Region A, obliterating statistical power and demonstrating why a locally adaptive procedure like the proposed Global Envelope Test (GET) is essential.

\section{Implementation of Algorithm \ref{strongenvalgorithm}}\label{app:implementation}

Algorithm \ref{strongenvalgorithm} can be efficiently implemented without the need to re-compute the measures for test statistics $T_j$ from scratch on each iteration. 
Namely, for the step-down envelope with the cont measure, the pointwise continuous ranks $C_j^i$ can be computed only once and
the continuous rank measure \eqref{eq:cont} can be straightforwardly computed for any required set of hypotheses using the $C_j^i$. 
This same approach also works for the rank, qdir and F-max measures, which are all infimal depths.

The area measure \eqref{eq:area_rank_measure} is defined using $R_j$ and $C_j^i$. Since $R_j = \lceil \min_i C_j^i \rceil$, the area measure can also be computed based on the $C_j^i$ computed in the beginning.

For the ERL measure, the pointwise ranks can be computed and sorted only once.
At each step of Algorithm \ref{strongenvalgorithm}, a hypothesis is dropped. 
The pointwise rank of the observation for the dropped hypothesis is $R_0^i$.
The sorted ranks for each function can now be updated by removing an element that is equal to $R_0^i$.

When the approximate method from \citet{MrkvickaEtal2022} is used for ERL, the following update is used on each function separately.
Recall that the approximation only tracks the count of the six most extreme pointwise ranks for each function. 
If the $R_0^i$ is larger than the largest rank in the auxiliary variables, then no updates on the counts are needed.
If the $R_0^i$ is equal to one of the ranks in the auxiliary variables, then the count of that rank is decreased by one. If this results in a zero count, then the auxiliary variables are recomputed fully.

\section{Additional description and results from the simulation study}\label{suppl:simstudy}

Table \ref{tab:simstudy_means} gives specifications for all Models 1-10 used in the simulation study. 
Figures \ref{fig:simstudy_fwer_Gaussian}-\ref{fig:simstudy_power_logGaussian2} provide results from all Models 1-10, complementing Figures \ref{fig:simstudy_fwer} and \ref{fig:simstudy_power} of the main text.
In particular, Figures \ref{fig:simstudy_fwer_Gaussian}-\ref{fig:simstudy_fwer_n} provide the results for empirical FWER, similar to the panels shown in Figures \ref{fig:simstudy_fwer}.
Figures \ref{fig:simstudy_power_Gaussian}-\ref{fig:simstudy_power_logGaussian2} provide the power results, similar to the panels shown in Figures \ref{fig:simstudy_power}.

\begin{table}[h!]
    \caption{\label{tab:simstudy_means} Ten models of the simulation study: specified group effects $Z_v(t)$, $t\in [0,1]$, for individuals $v$ in group 1 or 2, parameters, proportions of the true nulls ($m_0 / m$), and the standard deviation $\sigma_Z$ of the Gaussian error. 
    For illustration of the group effects, see the mean functions in \citet[][Figure A1]{MrkvickaMyllymaki2024}.}
\centering
\resizebox{\columnwidth}{!}{ 
\begin{tabular}{l|lllll} \hline
Model &    1   & 2-4   & 5-7   & 8-10  \\\hline 
Name  &   No diff. &  Abrupt & Fluctuating &  Soft \\
$Z_v(t)$ for $v \in$ group 1  &  0 & 0 & 0 & 0 \\
$Z_v(t)$ for $v \in$ group 2  &  0 &  $1$ for $t \in (0,a)$,  & $\max(0, \sin(b t)-t)$ & $\exp\{-((t - 0.5)^2/(2 \sigma^2)\}$  \\
            &    &  $0$ otherwise          &                        & for $t \in [0.5-3\sigma, 0.5+3\sigma]$, \\
            &    &                         &                        & $0$ otherwise \\
Parameters &  &  $a = 0.05$, $0.20, 0.50$ &  $b=20, 50, 100$ & $\sigma = 0.01, 0.025, 0.05$ &  \\
$m_0/m$     &  1 & 0.95, 0.8, 0.5 & 0.65, 0.67, 0.67 & 0.94, 0.85, 0.70  \\
$\sigma_Z$  &  1.5 &  2.1 &  0.9 & 0.9 \\\hline
    \end{tabular}
    }
\end{table}


\begin{figure}[h!]
    \centering
    \includegraphics[width=\linewidth]{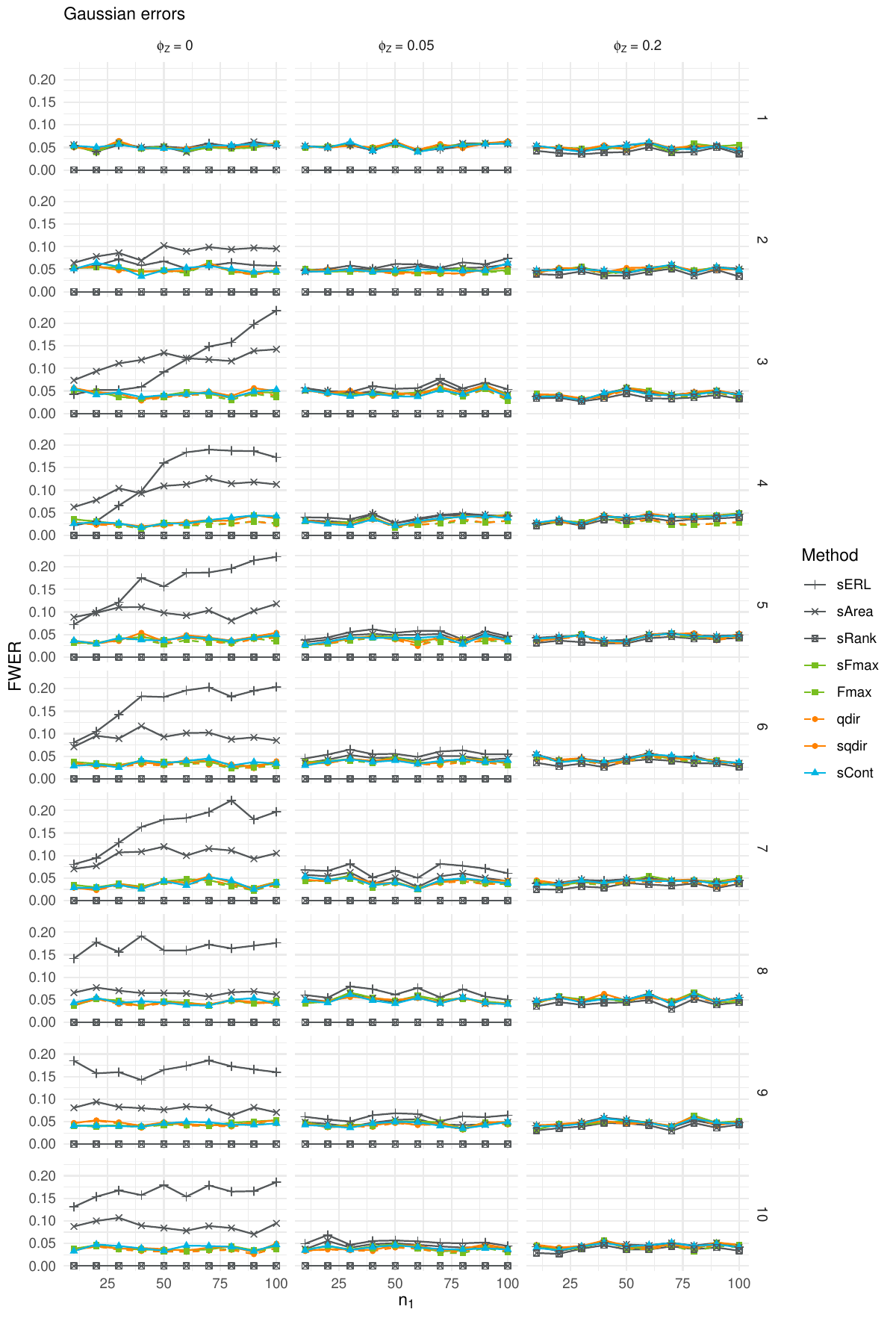}
    \caption{Empirical FWER 
    for the different methods (curves) and numbers of observations in the first group, $n_1$ (x-axis), for Models 1-10 (rows) with Gaussian errors and range of correlation $\phi_Z = 0$, $0.05$ or $0.2$ (columns). 
    In all cases, $n_2=3 n_1$, $s=999$ and $m=200$.
    }
    \label{fig:simstudy_fwer_Gaussian}
\end{figure}

\begin{figure}[h!]
    \centering
    \includegraphics[width=\linewidth]{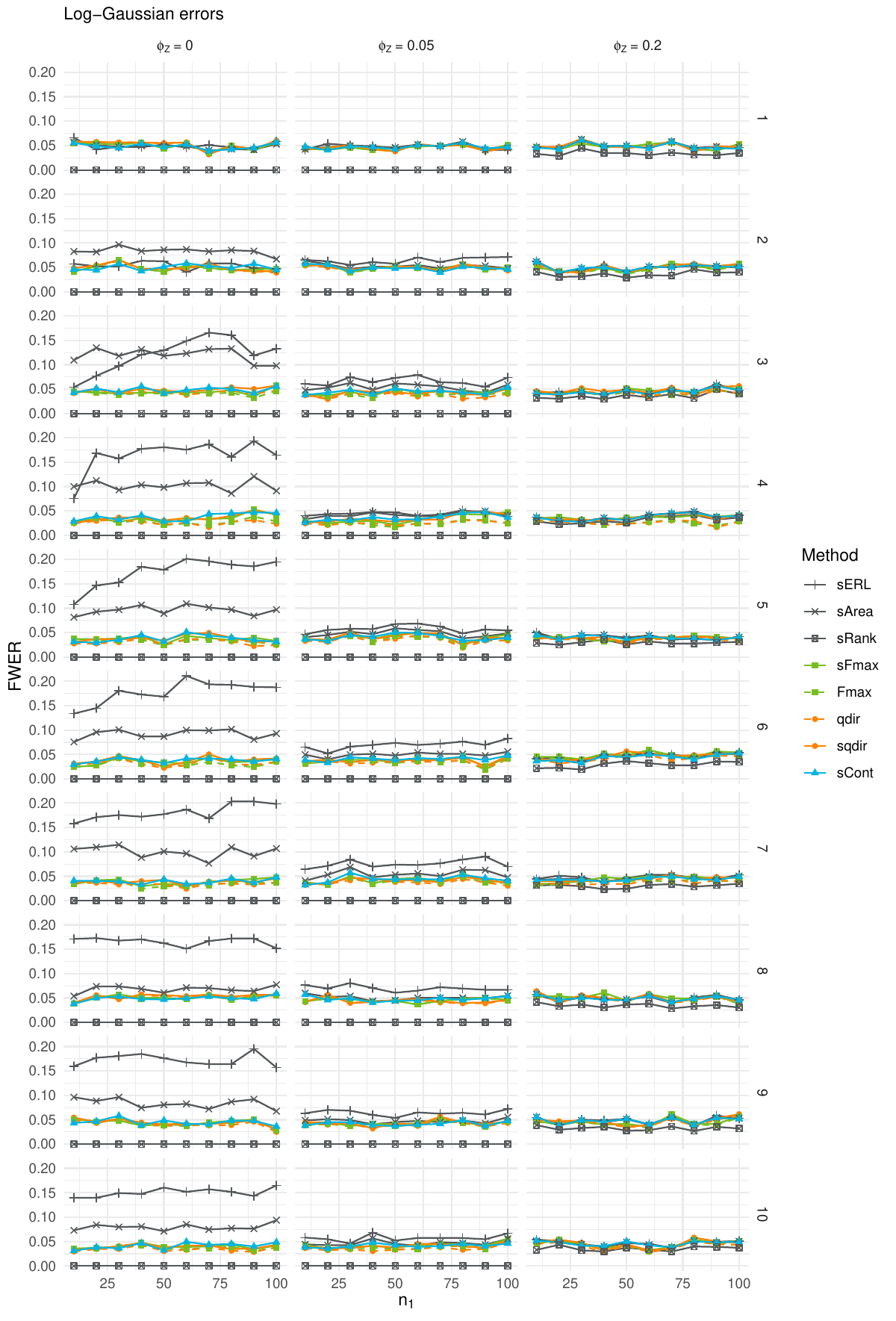}
    \caption{Empirical FWER 
    for the different methods (curves) and numbers of observations in the first group, $n_1$ (x-axis), for Models 1-10 (rows) with log Gaussian errors and range of correlation $\phi_Z = 0$, $0.05$ or $0.2$ (columns). 
    In all cases, $n_2=3 n_1$, $s=999$ and $m=200$.
    }
    \label{fig:simstudy_fwer_logGaussian}
\end{figure}

\begin{figure}[h!]
    \centering
    \includegraphics[width=\linewidth]{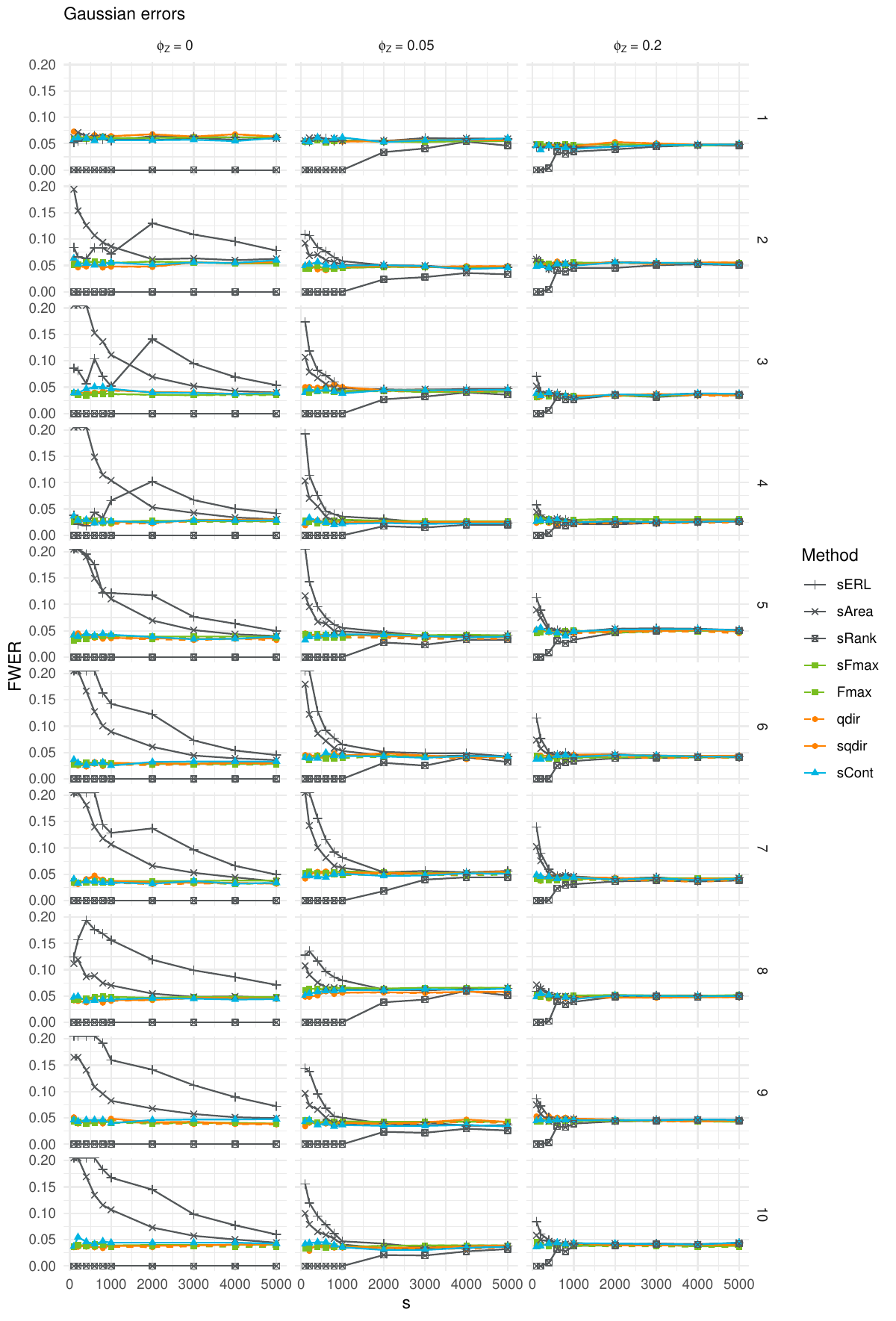}
    \caption{Empirical FWER 
    for the different methods (curves) and numbers of resamples $s$ (x-axis) for Models 1-10 (rows) with Gaussian errors and range of correlation $\phi_Z = 0$, $0.05$ or $0.2$ (columns). 
    The numbers of observations are $n_1 = 30$ and $n_2 = 90$, and $m=200$.
    }
    \label{fig:simstudy_fwer_n}
\end{figure}

\begin{figure}[h!]
    \centering
    \includegraphics[width=\linewidth]{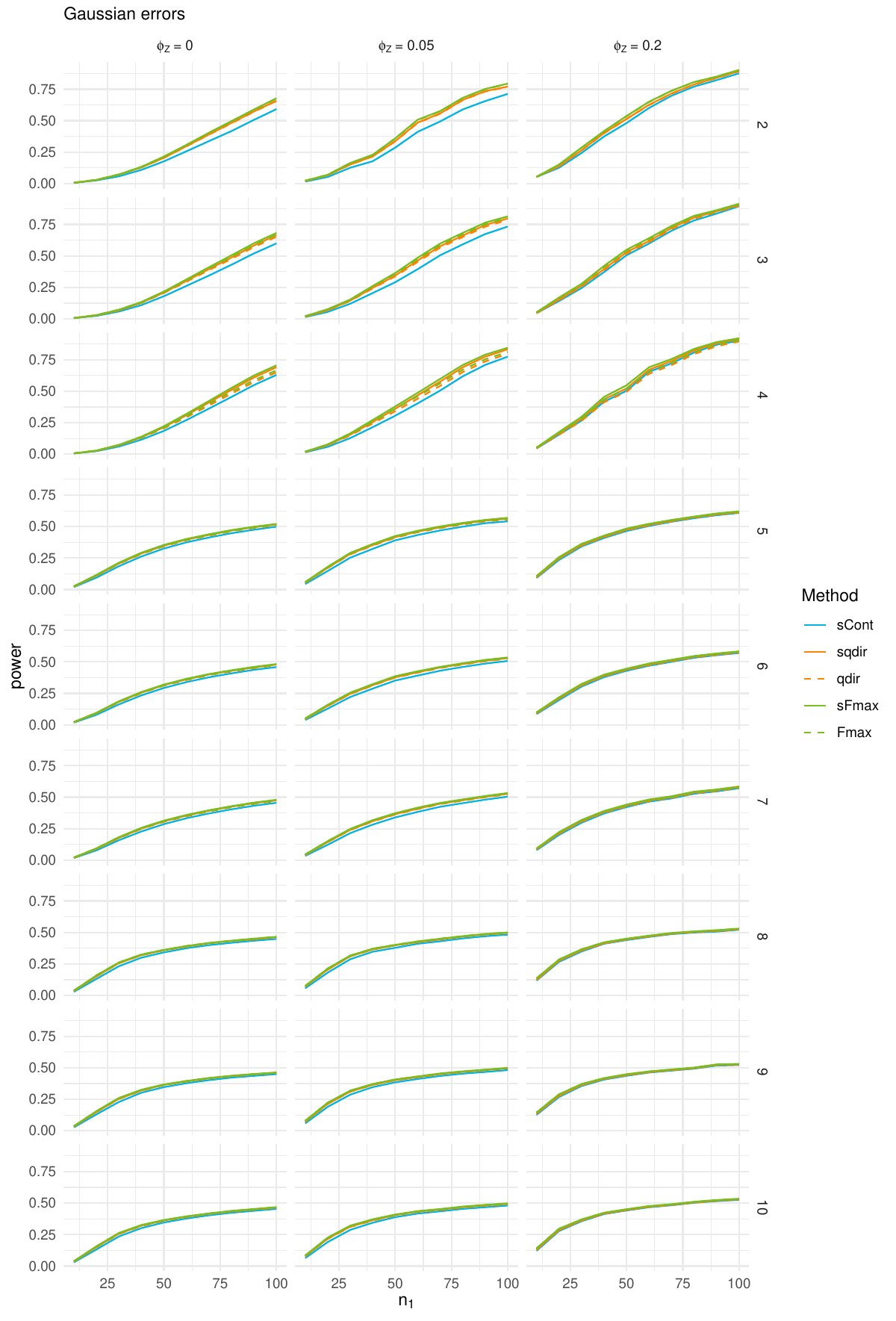}
    \caption{Empirical power (mean proportion of the alternative hypotheses rejected) 
    for the different methods (curves) and numbers of observations in the first group, $n_1$ (x-axis), for Models 2-10 (rows) with Gaussian errors and range of correlation $\phi_Z = 0$, $0.05$ or $0.2$ (columns).  
    The number of observations in the second group is $n_2=3 n_1$,
    $s=999$ and $m=200$.
    }
    \label{fig:simstudy_power_Gaussian}
\end{figure}

\begin{figure}[h!]
    \centering
    \includegraphics[width=\linewidth]{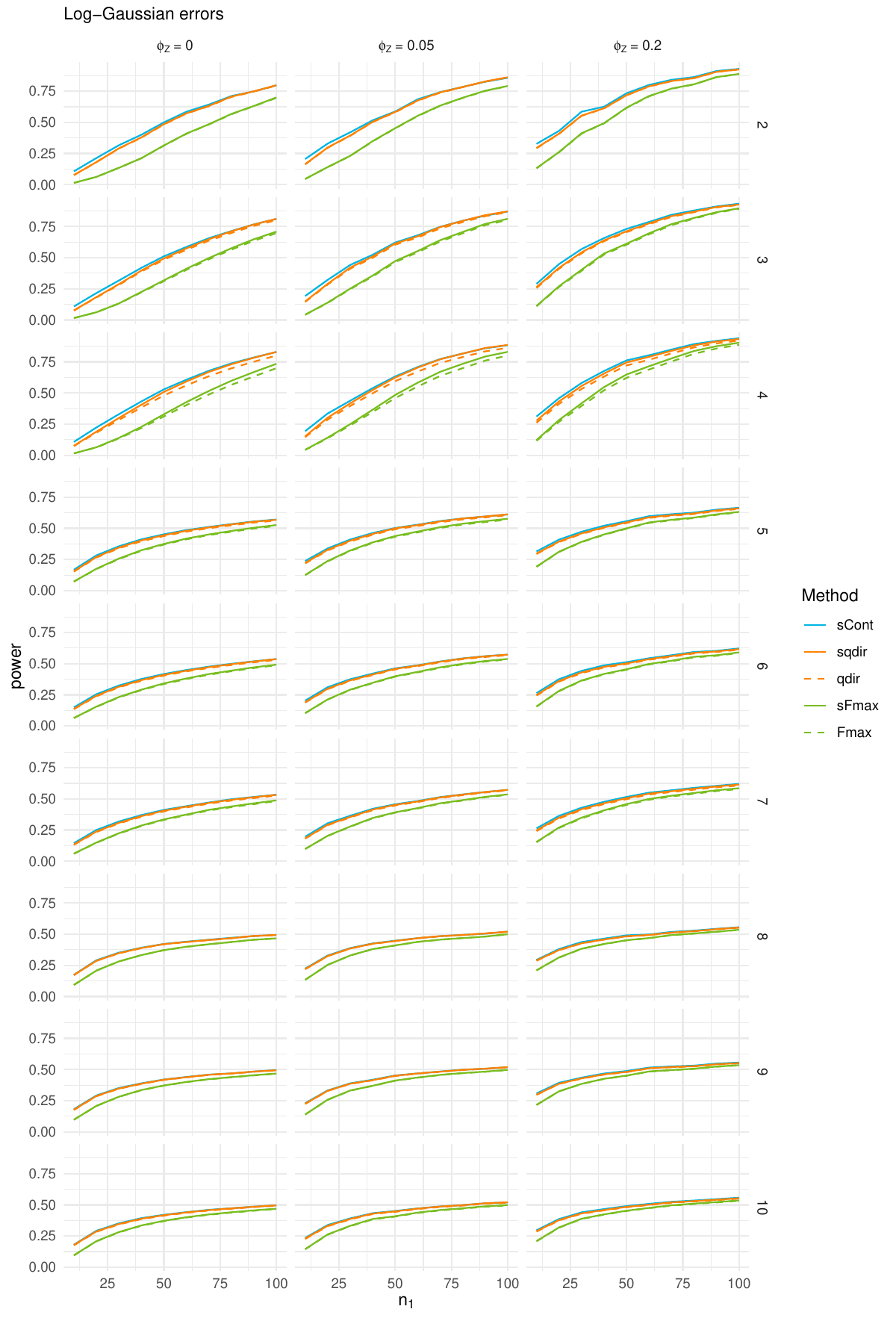}
    \caption{Empirical power (mean proportion of the alternative hypotheses rejected) 
    for the different methods (curves) and numbers of observations in the first group, $n_1$ (x-axis), for Models 2-10 (rows) with log Gaussian errors and range of correlation $\phi_Z = 0$, $0.05$ or $0.2$ (columns). 
    The number of observations in the second group is $n_2=3 n_1$,
    $s=999$ and $m=200$.
    }
    \label{fig:simstudy_power_logGaussian}
\end{figure}

\begin{figure}[h!]
    \centering
    \includegraphics[width=\linewidth]{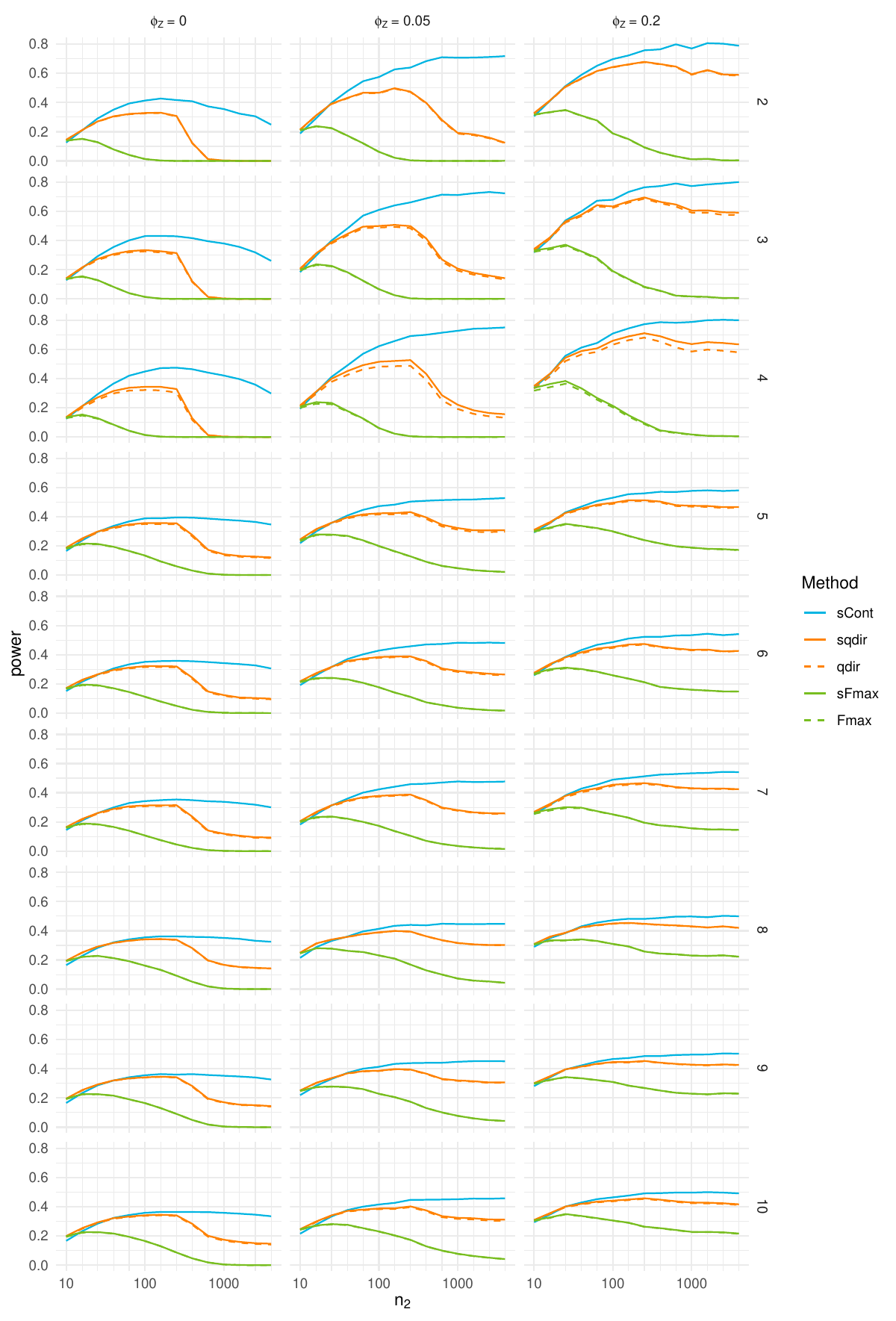}
    \caption{Empirical power (mean proportion of the alternative hypotheses rejected) 
    for the different methods (curves) and numbers of observations in the second group, $n_2$ (x-axis), for Models 2-10 (rows) with log Gaussian errors and range of correlation $\phi_Z = 0$, $0.05$ or $0.2$ (columns). 
    The number of observations in the first group is $n_1=10$,
    $s=999$ and $m=200$.
    }
    \label{fig:simstudy_power_logGaussian2}
\end{figure}

\end{document}